\documentclass[journal]{IEEEtran}
\pdfoutput=1
\usepackage{algorithmic}
\usepackage{algorithm}
\usepackage{array}

\usepackage{subfigure}
\usepackage{amsmath}
\usepackage{amsfonts}
\usepackage{amsthm}

\usepackage{textcomp}
\usepackage{stfloats}
\usepackage{url}
\usepackage{verbatim}
\usepackage{graphicx}
\usepackage{cite}
\usepackage{bm}
\usepackage{soul}
\usepackage{arydshln}
\usepackage{booktabs}
\usepackage{pifont}
\usepackage{amssymb}
\usepackage{tabularx}
\usepackage{microtype}
\usepackage{stfloats}
\usepackage{epstopdf}
\usepackage{makecell}
\usepackage{diagbox}
\usepackage{multicol}

\allowdisplaybreaks[4]
\usepackage[implicit=false]{hyperref}
\usepackage{tikz,xcolor}

\newtheorem{lemma}{Lemma}
\newtheorem{remark}{Remark}
\newtheorem{corollary}{Corollary}
\newtheorem{proposition}{Proposition}

\begin{document}
\definecolor{lime}{HTML}{A6CE39}
\DeclareRobustCommand{\orcidicon}{%
    \begin{tikzpicture}
    \draw[lime, fill=lime] (0,0)
    circle [radius=0.16]
    node[white] {{\fontfamily{qag}\selectfont \tiny ID}};    \draw[white, fill=white] (-0.0625,0.095)
    circle [radius=0.007];    \end{tikzpicture}
    \hspace{-2mm}}
\foreach \x in {A, ..., Z}{%
    \expandafter\xdef\csname orcid\x\endcsname{\noexpand\href{https://orcid.org/\csname orcidauthor\x\endcsname}{\noexpand\orcidicon}}
    }
\hypersetup{hidelinks}
\newcommand{\orcidauthorA}{0000-0001-6215-6703}

\title{Performance Bounds and Optimization for CSI-Ratio based Bi-static Doppler Sensing\\ in ISAC Systems}
\author{Yanmo Hu,~\IEEEmembership{Student~Member,~IEEE}, Kai Wu,~\IEEEmembership{Member,~IEEE},\\J. Andrew Zhang,~\IEEEmembership{Senior~Member,~IEEE}, Weibo Deng, Y. Jay Guo,~\IEEEmembership{Fellow,~IEEE}\vspace{-0.25em}
\thanks{
Yanmo Hu and Weibo Deng are with the Key Laboratory of Marine Environmental Monitoring and Information Processing, Ministry of Industry and Information Technology, and also with the School of Electronic and Information Engineering, Harbin Institute of Technology. Email: 19B905008@stu.hit.edu.cn; dengweibo@hit.edu.cn. \emph{(Corresponding author: Weibo Deng.)}

Kai Wu, J. A. Zhang and Y. J. Guo are with Global Big Data Technologies Centre, the University of Technology Sydney. Email: \{Kai.Wu; Andrew.Zhang; Jay.Guo\}@uts.edu.au.

}}

\maketitle

\begin{abstract}
Bi-static sensing is crucial for exploring the potential of networked sensing capabilities in integrated sensing and communications (ISAC).
However, it suffers from the challenging clock asynchronism issue.
CSI ratio-based sensing is an effective means to address the issue.
Its performance bounds, particular for Doppler sensing, have not been fully understood yet.
This work endeavors to fill the research gap.
Focusing on a single dynamic path in high-SNR scenarios, we derive the closed-form CRB.
Then, through analyzing the mutual interference between dynamic and static paths, we simplify the CRB results by deriving {close} approximations, further unveiling new insights of the impact of numerous physical parameters on Doppler sensing.
Moreover, utilizing the new CRB and analyses, we propose novel waveform optimization strategies for noise- and interference-limited sensing scenarios, which are also empowered by closed-form and efficient solutions.
Extensive simulation results are provided to validate the preciseness of the derived CRB results and analyses,
with the aid of the maximum-likelihood estimator.
The results also demonstrate the substantial enhanced Doppler sensing accuracy and the sensing capabilities for low-speed target achieved by the proposed waveform design.
\end{abstract}


\section{Introduction}\label{Section_Introduction}

\IEEEPARstart{I}{ntegrated}
sensing and communications (ISAC) has established its critical role in numerous next-generation communication networks,
such as mobile networks,
WiFi networks and vehicular networks etc \cite{9585321, 9724258, 9804861, 9393557, pegoraro2023jump, 9606831}.
Major sensing types extensively discussed in ISAC include single-node sensing with sensing transceivers co-located and bi-static sensing with sensing transmitters far separated from receivers.
While single-node sensing enjoy clock coherence between transceiver channels,
it requires full-duplex technique,
which is not mature yet \cite{9848428}.
Bi-static sensing can eliminate the requirement of full-duplexing techniques
and has the potential to utilize the networked sensing capability to deliver enhanced sensing performances \cite{10207009, 9606921}.
Take the emerging perceptive mobile networks (PMNs) \cite{8827589} for an example.
Bi-static sensing at a base station using the link signals from almost ubiquitous user ends can enjoy unprecedented spatial diversity
and information fusion gain for sensing.
However, such bi-static sensing is not without challenges. Due to
different clock oscillators,
between transmitter and receiver, bi-static sensing suffers from the clock asynchronism issue \cite{9848428, 8443427, 9617146}.
It can lead to sensing ambiguities and
must be solved for fully enjoying the potential networked sensing ISAC.

Several types of methods have been proposed to solve the clock asynchronism issue in ISAC.
Since the antennas at a communication receiver typically share a common clock source,
the signals received by each antenna have the same timing and phase nuisances caused by clock asynchronism.
Leveraging this characteristic,
the cross-antenna cross-correlation (CACC) methods were proposed in \cite{9349171, 2018Widar},
which compensate clock asynchronism errors by taking
the cross-correlations of the
received signals over different antenna pairs.
However, the CACC method doubles the sensing parameters that need to be estimated and
increases the probability of detecting false dynamic targets (also known as ghost targets).
Moreover, the CACC method
requires a line-of-sight (LOS) path to ensure unambiguous sensing.
In practical scenarios, the LOS requirement may not be easy to be satisfied.

The CSI-ratio model \cite{123777_ACM, 9904500, 10188607},
a counterpart of CACC,
has
the ability to
alleviate the LOS requirement.
Different from CACC,
the CSI-ratio model calculates the ratio of signals among different antennas
to eliminate common errors caused by clock asynchronism.
In \cite{9904500},
the Doppler frequency of a single dynamic target is estimated by
employing the Mobius transformation,
which
constructs a
periodic
conformal map
between CSI ratio and Doppler frequency.
The Doppler frequency can be estimated by calculating the
period of the CSI ratio.
Furthermore,
\cite{10188607}
introduces the MUSIC algorithm to the CSI-ratio-based sensing.
Compared with the CACC-based MUSIC method,
the CSI-ratio-based MUSIC has better performance in estimating delay and Doppler
in both LOS and NLOS scenarios \cite{123777_ACM, 9848428}.
Therefore, as an elegant formulation,
the CSI ratio model allows simple and efficient parameter estimation.
It also has the capability of
removing additional common signal distortions across antennas such as varying receiver gains \cite{123777_ACM},
remarkably standing out from algorithms for solving the clock
asynchronism issue.

Despite advancements of the CSI-ratio model,
assessing its optimality remains unsolved
and signal optimization lacks of reliable performance metrics.
In parameter estimation,
the classical
Cram\'{e}r-Rao bounds (CRB),
as the lower bound of the variance of an unbiased estimator,
is
an effective metric for assessing the optimality of an estimator and signal design \cite{FSSPVI}.
However,
CRB for the CSI-ratio model,
and more generally for bi-static sensing,
is not known yet,
although limited works
have studied the CRB for ISAC without clock asynchronism \cite{9468975, 8561147, 10318068, 10288116}.
For downlink sensing,
\cite{9468975} optimizes
the subcarrier and OFDM symbol indexes
to minimize delay and Doppler frequency CRBs, respectively.
Additionally,
\cite{10288116}
minimizes the closed-form CRB
to improve delay estimation accuracy in 5G NR networks \cite{8254568}.
In
\cite{10318068},
an uplink waveform optimization problem is solved with the
synchronized clock
assumption.
An algorithm is proposed to
jointly optimize the subcarrier indexes and the power allocation to enhance the estimation accuracy.
These studies highlight the importance of CRB in performance
analysis and optimization in ISAC.
However,
it is very challenging to derive meaningful CRB results for bi-static sensing with
clock asynchronism, due to the multiple unknown time-varying variables of offsets.

To contribute to filling the research gap in the literature,
this paper is dedicated to analyzing and optimizing the Doppler CRB for CSI ratio-based bi-static sensing in ISAC.
We start with deriving the analytical CRB under high-signal-to-noise-ratio (SNR) scenarios.
Then,
utilizing the interference pattern between dynamic and static paths, we provide an approximated CRB,
which not only closely approaches the exact CBR but also reveals impact of critical physical quantities on Doppler sensing.
Moreover,
employing the derived CBR and the unveiled insights through CRB analysis,
we propose new waveform design to optimize the Doppler sensing performance in CSI ratio-based bi-static sensing.
The major contributions are highlighted as follows:
\begin{itemize}

\item
We derive the closed-form CRB for Doppler sensing employing the CSI-ratio model for bi-static sensing in high-SNR scenarios for a single dynamic path.
Through analyzing the impact of mutual impact between dynamic and static paths,
we manage to simplify the closed-form CRB by deriving a {close} CRB approximation.

\item
We perform extensive analysis employing the derived CRB results,
unveiling new and interesting insights into the impact of different parameters,
such as average power of static and dynamic paths,
incident angles of static and dynamic paths,
and noise power etc.,
on Doppler sensing performances.

\item
We provide new waveform design to optimize Doppler sensing performance in different bi-static sensing scenarios,
utilizing the preceding CRB results and analyses.
For noise-limited cases, where interference between dynamic and static paths can be ignored,
we derive the optimal waveform for Doppler sensing.
Moreover, for interference-limited sensing,
we develop a novel waveform optimization strategy to enhance the region of sensible Doppler frequencies,
particularly for small values.
The strategy also entails an optimal trade-off between interference- and noise-limited Doppler sensing.

\end{itemize}

Extensive simulation results are provided to demonstrate the validity of the CRB results and analyses.
A maximum likelihood estimator is used to validate the derived CRBs.
Impact of different parameters on CRBs is also simulated,
to consolidate the analyses.
We also validate the preciseness of the new derivations and analyses, as well as the optimality of the proposed waveform design.

The rest of the paper is organized as follows.
Section \ref{Section_II} introduces the signal model of the CSI-ratio.
The derivation of CRB of the CSI-ratio model and the analysis of the Doppler CRB are presented in Section \ref{Section_III}.
Section \ref{Section_IV}
proposes two algorithms to enhance the model accuracy by optimizing indexes of the OFDM symbols.
Simulation results are provided in Section \ref{Section_simulation_results},
and conclusions are drawn in Section \ref{Section_Conclusion}.

$Notation$:
$\odot $ denotes the Hadamard product;
if $\mathbf{a} \in \mathbb{C}^{N\times 1}$ is a vector,
$\left| \mathbf{a} \right|$,
${{\left| \mathbf{a} \right|}^{2}}$,
and
${{\mathbf{a}}^{2}}$
denote the absolute value, the square of the absolute value, and the square value of each element in $\mathbf{a}$;
${{\left[ \mathbf{a}  \right]}_{n}}$ denotes the $n$-th element of vector $\mathbf{a}$,
and ${{\left[ \mathbf{A}  \right]}_{m,n}}$ denotes the $\left(m,n \right)$-th element of matrix $\mathbf{A}$;
if $\mathbf{a}, \mathbf{b} \in \mathbb{C}^{N\times 1}$ are vectors,
${\mathbf{a}}/{\mathbf{b}}$ denotes the division of the corresponding elements in two vectors,
where ${{\left[ {\mathbf{a}}/{\mathbf{b}} \right]}_{n}}={{{a}_{n}}}/{{{b}_{n}}}$, $n = 0, \cdots, N - 1$;
$\angle a$ denotes the argument of the complex number $a \in \mathbb{C}$;
$\text{round}\left\{ a \right\}$ denotes the rounding operation for $a \in \mathbb{R}$;
${C_{n}^{r}}=\frac{n!}{r!\left( n-r \right)!}$
denotes the combinatorial number;
${{\mathcal{F}}}\left\{\cdot\right\}$ is the discrete Fourier transform (DFT) operator;
$\text{sinc}\left(x\right) = {\sin \left( \pi x \right)}/{\left(\pi x\right)}$
is the sinc function.

\section{CSI-Ratio Model}\label{Section_II}

In PMN ISAC,
uplink sensing generally suffers from clock asynchronism issues \cite{9848428, 9585321}.
Specifically, the clock offset
between transmitter and receiver incurs TMO and CFO, which further leads to
the Doppler and delay ambiguity.
Thus, the TMO and CFO compensation is essential for uplink sensing.
One effective method for eliminating CFO and TMO is based on the CSI ratio between different receiving antennas.
Additionally,
we assume that the TMO and CFO remain constants within each OFDM symbol \cite{9848428},
and we use $\tau_{o, k}$ and $\beta_{k}$ to denote the TMO and CFO of the $k$-th OFDM symbol.

In 5G NR, the OFDM symbols in physical uplink shared channel (PUSCH) can be designed for improving the sensing performance \cite{5GNR213, 8928165}.
Therefore, we consider a general case that the interval of adjacent OFDM symbols for sensing can be unequal,
and the symbol interval is an integer multiple of $T_0$.
We use $K_\text{all}$ to represent the number of the available OFDM symbols.
Suppose that $K$ out of $K_\text{all}$ OFDM symbols are allocated for the uplink sensing, and the index of the $k$-th sensing symbol is $\varphi_k \in \mathbb{Z}$,
where
$0 \le \varphi_k \le K_\text{all} - 1$
with
$k = 0, \cdots, K - 1$.
The sensing receiver at the BS uses
an uniform linear array.
Let $\Delta f$ denote the subcarrier interval.
Moreover, we consider that there are
$L_s$ static objects (with AOA $\theta_{s, \ell}$ {and delay $\tau_{s, \ell}$})
and
one dynamic
path associated with a moving target
(with the Doppler frequency $f_{d}$, AOA $\theta_{d}$, {and delay $\tau_{d}$}).
Accordingly,
the CSI {for the $p$-th subcarrier} of the $k$-th OFDM symbol and the $m$-th antenna of the BS can be expressed as \cite{9585321}:
\begin{align}\label{signal_model_abddd}
\nonumber y_{p,k}^{\left( m \right)}= {} & {{e}^{j{{\beta }_{k}}}}{{e}^{j2\pi p\Delta f{{\tau }_{o,k}}}}\left( {{\xi }_{d}}{{e}^{j2\pi p\Delta f{{\tau }_{d}}}}{{e}^{j\frac{2\pi }{\lambda }{{d}_{m}}\sin {{\theta }_{d}}}}{{e}^{j2\pi {{\varphi }_{k}}{{f}_{d}}{{T}_{0}}}} \right. \\
 & \left.  +\sum\limits_{\ell =0}^{{{L}_{s}}-1}{{{\xi }_{s,\ell }}{{e}^{j2\pi p\Delta f{{\tau }_{s,\ell }}}}{{e}^{j\frac{2\pi }{\lambda }{{d}_{m}}\sin {{\theta }_{s,\ell }}}}} \right)+z_{p,k}^{\left( m \right)},
\end{align}
where
$\xi_{d}$ and $\xi_{s, \ell}$
denote the combined value of
the complex-valued reflection coefficient and the propagation
attenuation for the dynamic target and the $\ell$-th static object, respectively;
$\lambda$ denotes the wavelength;
$d_m$ denotes the position of the $m$-th antenna;
$z_{p,k}^{\left( m \right)}\sim \mathcal{C}\mathcal{N}\left( 0,\sigma _{n}^{2} \right)$
is the Gaussian noise.

The CSI ratio between two different receiving antennas
can
effectively eliminate the impact of the clock offset \cite{123777_ACM}.
Therefore, we consider a simple
two-element array, $M = 2$, with $d$ denoting the spacing between the two elements
and $a\left(\theta_d\right) = {{e}^{j\frac{2\pi }{\lambda }d\sin {{\theta_d }}}}$
denoting the phasor
caused by the phase difference
between two array elements.
The CSI ratio
can be calculated as:
\begin{align}\label{aaa}
\nonumber r\left( k \right) &={y_{k}^{\left( 1 \right)}}/{y_{k}^{\left( 0 \right)}}\; \\
 & =\frac{a\left( {{\theta }_{d}} \right)\rho _{1}^{\left( p \right)}{{e}^{j2\pi {{\varphi }_{k}}{{f}_{d}}{{T}_{0}}}}  +\rho _{0}^{\left( p \right)}+\left( {\tilde{z}_{p,k}^{\left( 1 \right)}}/{h_{s,0}^{\left( p \right)}} \right)}{\rho _{1}^{\left( p \right)}{{e}^{j2\pi {{\varphi }_{k}}{{f}_{d}}{{T}_{0}}}}  +1+\left( {\tilde{z}_{p,k}^{\left( 0 \right)}}/{h_{s,0}^{\left( p \right)}} \right)},
\end{align}
where
$\xi _{d}^{\left( p \right)}\triangleq {{\xi }_{d}}{{e}^{j2\pi p\Delta f{{\tau }_{d}}}}$,
$\xi _{s,\ell }^{\left( p \right)}\triangleq {{\xi }_{s,\ell }}{{e}^{j2\pi p\Delta f{{\tau }_{s,\ell }}}}$,
$\tilde{z}_{p,k}^{\left( 0 \right)}\triangleq{{e}^{-j{{\beta }_{k}}}}{{e}^{-j2\pi p\Delta f{{\tau }_{o,k}}}}z_{p,k}^{\left( 0 \right)}$,
$\tilde{z}_{p,k}^{\left( 1 \right)}\triangleq{{e}^{-j{{\beta }_{k}}}}{{e}^{-j2\pi p\Delta f{{\tau }_{o,k}}}}z_{p,k}^{\left( 1 \right)}$,
and four key intermediate variables are defined as
\begin{align}\label{aaa_variables_definition}
\nonumber \rho _{0}^{\left( p \right)}=\frac{h_{s,1}^{\left( p \right)}}{h_{s,0}^{\left( p \right)}}, & \ \ \rho _{1}^{\left( p \right)}=\frac{\xi _{d}^{\left( p \right)}}{h_{s,0}^{\left( p \right)}} \\
h_{s,0}^{\left( p \right)}=\sum\limits_{\ell =0}^{{{L}_{s}}-1}{\xi _{s,\ell }^{\left( p \right)}}, & \ \ h_{s,1}^{\left( p \right)}=\sum\limits_{\ell =0}^{{{L}_{s}}-1}{\xi _{s,\ell }^{\left( p \right)}{{e}^{j\frac{2\pi }{\lambda }d\sin {{\theta }_{s,\ell }}}}}.
\end{align}
As is typical \cite{9904500}, we employ the uplink signal from \emph{the same subcarrier in each OFDM symbol for Doppler sensing}.
The analysis in this work can be readily applicable to any subcarrier.
Hence,
we omit the subcarrier index $p$ of all variables,
and
$\rho _{0}^{\left( p \right)}$,
$\rho _{1}^{\left( p \right)}$,
$h_{s,0}^{\left( p \right)}$,
$h_{s,1}^{\left( p \right)}$,
$\xi _{d}^{\left( p \right)}$,
$\xi _{s,\ell }^{\left( p \right)}$,
$\tilde{z}_{p,k}^{\left( 0 \right)}$,
$\tilde{z}_{p,k}^{\left( 1 \right)}$
are re-expressed as
${{\rho }_{0}}$,
${{\rho }_{1}}$,
${{h}_{s,0}}$,
${{h}_{s,1}}$,
${{\xi }_{d}}$,
${{\xi }_{s,\ell }}$,
$\tilde{z}_{k}^{\left( 0 \right)}$,
$\tilde{z}_{k}^{\left( 1 \right)}$,
respectively.

The terms ${{h}_{s,0}}$ and ${{h}_{s,1}}$
denote the static part in the {received signals} from the $0$-th and $1$-st antennas, respectively.
Both the
numerator and denominator of (\ref{aaa})
have terms related to Gaussian noises, and the CSI ratio expression is nonlinear.
Therefore,
the probability density function of (\ref{aaa}) is too complicated to be derived for calculating the CRB \cite{Math_ratio_2013}.
As a typical strategy, we introduce a legitimate assumption, i.e., high SNR,
to simplify the derivation
of the CRB
and then evaluate the impact of the assumption later in Section \ref{Section_simulation_BBBB}.
Considering the \emph{high-SNR case}, the CSI ratio can be approximated using {the} Taylor {series} expansion, as given by:
\begin{align}\label{R_approximation}
\nonumber r\left( k \right)\approx {}& \frac{a\left( {{\theta }_{d}} \right){{\rho }_{1}}{{e}^{j2\pi {{\varphi }_{k}}{{f}_{d}}{{T}_{0}}}}+{{\rho }_{0}}}{{{\rho }_{1}}{{e}^{j2\pi {{\varphi }_{k}}{{f}_{d}}{{T}_{0}}}}+1}+\frac{ {\tilde{z}_{k}^{\left( 1 \right)}}/{{{h}_{s,0}}} }{{{\rho }_{1}}{{e}^{j2\pi {{\varphi }_{k}}{{f}_{d}}{{T}_{0}}}}+1} \\
 & -\frac{a\left( {{\theta }_{d}} \right){{\rho }_{1}}{{e}^{j2\pi {{\varphi }_{k}}{{f}_{d}}{{T}_{0}}}}+{{\rho }_{0}}}{{{\left( {{\rho }_{1}}{{e}^{j2\pi {{\varphi }_{k}}{{f}_{d}}{{T}_{0}}}}+1 \right)}^{2}}}\left( {\tilde{z}_{k}^{\left( 0 \right)}}/{{{h}_{s,0}}} \right).
\end{align}
Note that {the} approximation holds when the noise term
{${\tilde{z}_{k}^{\left( 1 \right)}}/{{{h}_{s,0}}} $
and
${\tilde{z}_{k}^{\left( 0 \right)}}/{{{h}_{s,0}}} $}
in (\ref{aaa}) are much smaller than other components.
The
more specific
conditions of the approximation can be established as follows.

\vspace{0.3em}

\begin{lemma}\label{fundamental_assumption_high_SNR}
If $\left| {{\xi }_{d}} \right|<\min \left\{ \left| {{h}_{s,0}} \right|,\left| {{h}_{s,1}} \right| \right\}$,
the approximation in (\ref{R_approximation}) holds when:
\begin{align}
\min \left\{ \left| {{h}_{s,0}} \right|,\left| {{h}_{s,1}} \right| \right\}-\left| {{\xi }_{d}} \right|\gg {{\sigma }_{n}}.
\end{align}
If $\left| {{\xi }_{d}} \right|>\max \left\{ \left| {{h}_{s,0}} \right|,\left| {{h}_{s,1}} \right| \right\}$,
the approximation holds when:
\begin{align}
\left| {{\xi }_{d}} \right|-\max \left\{ \left| {{h}_{s,0}} \right|,\left| {{h}_{s,1}} \right| \right\}\gg {{\sigma }_{n}}.
\end{align}
\end{lemma}

\vspace{0.3em}

For convenience of incorporating all sensing symbols in the CRB analysis, we introduce the matrix form of the CSI ratio below.
With the symbol interval $T_0$, the Doppler steering vector can be given by
\[\mathbf{d}\left( {{f}_{d}} \right)={{\left[ {{e}^{j2\pi {{\varphi }_{0}}{{f}_{d}}{{T}_{0}}}},\cdots ,{{e}^{j2\pi {{\varphi }_{K-1}}{{f}_{d}}{{T}_{0}}}} \right]}^{T}}\in {{\mathbb{C}}^{K\times 1}}.\]
Note that $\varphi_k$ is the symbol index and may not be evenly distributed in our analysis.
Moreover, let the vectors
${{\mathbf{\tilde{z}}}^{\left( 0 \right)}}={{\left[ \tilde{z}_{0}^{\left( 0 \right)},\cdots ,\tilde{z}_{K-1}^{\left( 0 \right)} \right]}^{T}}h_{s,0}^{-1}\in {{\mathbb{C}}^{K\times 1}}$
and
${{\mathbf{\tilde{z}}}^{\left( 1 \right)}}={{\left[ \tilde{z}_{0}^{\left( 1 \right)},\cdots ,\tilde{z}_{K-1}^{\left( 1 \right)} \right]}^{T}}h_{s,0}^{-1}\in {{\mathbb{C}}^{K\times 1}}$
denote the noise vectors of the ${0}$-th and the ${1}$-st antennas with respect to
$K$ {numbers of} OFDM symbols, respectively.
Based on (\ref{R_approximation}),
the vector of CSI ratios over all sensing symbols can be written as:
\begin{align}\label{asdlifhjaolifhaoweifh}
\nonumber   \mathbf{r}={}& \underbrace{\frac{a\left( {{\theta }_{d}} \right){{\rho }_{1}}\mathbf{d}\left( {{f}_{d}} \right)+{{\rho }_{0}}{{\mathbf{1}}_{K\times 1}}}{{{\rho }_{1}}\mathbf{d}\left( {{f}_{d}} \right)+{{\mathbf{1}}_{K\times 1}}}}_{\bm{\chi }} \\
\nonumber & +\underbrace{\frac{{{{\mathbf{\tilde{z}}}}^{\left( 1 \right)}}}{{{\rho }_{1}}\mathbf{d}\left( {{f}_{d}} \right)+{{\mathbf{1}}_{K\times 1}}}-\frac{a\left( {{\theta }_{d}} \right){{\rho }_{1}}\mathbf{d}\left( {{f}_{d}} \right)+{{\rho }_{0}}{{\mathbf{1}}_{K\times 1}}}{{{\left( {{\rho }_{1}}\mathbf{d}\left( {{f}_{d}} \right)+{{\mathbf{1}}_{K\times 1}} \right)}^{2}}}\odot {{{\mathbf{\tilde{z}}}}^{\left( 0 \right)}}}_{\mathbf{z}} \\
\triangleq {}& \bm{\chi }+\mathbf{z}.
\end{align}

From
(\ref{R_approximation}),
we can verify that
$r\left(k\right)$
has a Gaussian distribution with the mean of $\bm{\chi}$
and variance of $\text{diag}\left\{\bm{\eta}\right\}$,
where
\[\bm{\eta }=\frac{\sigma _{n}^{2}}{{{\left| {{h}_{s,0}} \right|}^{2}}} \frac{{{\left| {{\rho }_{1}}\mathbf{d}\left( {{f}_{d}} \right)+{{\mathbf{1}}_{K\times 1}} \right|}^{2}}+{{\left| a\left( {{\theta }_{d}} \right){{\rho }_{1}}\mathbf{d}\left( {{f}_{d}} \right)+{{\rho }_{0}}{{\mathbf{1}}_{K\times 1}} \right|}^{2}}}{{{\left| {{\rho }_{1}}\mathbf{d}\left( {{f}_{d}} \right)+{{\mathbf{1}}_{K\times 1}} \right|}^{4}}}.\]
Hence,
$\mathbf{r}$
conforms to a joint Gaussian distribution,
as given by
$\mathbf{r}\sim\mathcal{C}\mathcal{N}\left( \bm{\chi },\text{diag}\left\{\bm{\eta}\right\} \right)$.
Next,
we employ these signal models to investigate the Doppler sensing performance using CSI ratio.

\section{CRB for CSI-Ratio-Based Sensing}\label{Section_III}
In this section, we first derive the closed-form CRB for CSI ratio-based Doppler sensing and then analyze the result
to highlight some insights.

\subsection{Deriving CRB}

The CRB derivation is based on the CSI ratio vector $\mathbf{r}$, as given in (\ref{asdlifhjaolifhaoweifh}).
Let us collect the
unknown parameters in (\ref{asdlifhjaolifhaoweifh}) by the following real vector:
\begin{align}\label{alpha_parameters_unknown}
\bm{\alpha }={{\left[ {f}_{d},{{{\theta }}_{d}} ,  \operatorname{Re}\left\{ {{\rho }_{0}} \right\},\operatorname{Im}\left\{ {{\rho }_{0}} \right\},\operatorname{Re}\left\{ {\rho }_{1} \right\},\operatorname{Im}\left\{ {\rho }_{1} \right\} \right]}^{T}}\in {{\mathbb{R}}^{6\times 1}},
\end{align}
where $\rho_0$ and $\rho_1$ are given in (\ref{aaa_variables_definition}).
As is typical, the CRB derivation is started with establishing the Fisher information matrix (FIM) \cite{FSSPVI}.
Let $\mathbf{F}\left( \bm{\alpha } \right) \in \mathbb{R}^{6 \times 6}$
denote the FIM.
The $\left(p, q\right)$-th entry can be given by:
\begin{align}\label{CRB_for_formal}
\nonumber  & {{\left[ \mathbf{F}\left( \bm{\alpha } \right) \right]}_{p,q}}=2\operatorname{Re}\left\{ {{\left( \frac{\partial \bm{\chi }}{\partial {{\left[ \bm{\alpha } \right]}_{p}}} \right)}^{H}}\text{dia}{{\text{g}}^{-1}}\left\{ \bm{\eta } \right\}\frac{\partial \bm{\chi }}{\partial {{\left[ \bm{\alpha } \right]}_{q}}} \right\} \\
 & +\text{tr}\left\{ \text{dia}{{\text{g}}^{-1}}\left\{ \bm{\eta } \right\}\frac{\partial \text{diag}\left\{ \bm{\eta } \right\}}{\partial {{\left[ \bm{\alpha } \right]}_{p}}}\text{dia}{{\text{g}}^{-1}}\left\{ \bm{\eta } \right\}\frac{\partial \text{diag}\left\{ \bm{\eta } \right\}}{\partial {{\left[ \bm{\alpha } \right]}_{q}}} \right\}.
\end{align}
The second term in (\ref{CRB_for_formal})
denotes the impact of noise.
Intuitively, when the SNR is high, the second term {becomes small}.
Thus, by applying Lemma \ref{fundamental_assumption_high_SNR}, the FIM can be approximately simplified into:
\begin{align}\label{CRB_for}
\mathbf{F}\left( \bm{\alpha } \right) \approx 2\operatorname{Re}\left\{ {{\left( \frac{\partial \bm{\chi }}{\partial {{\bm{\alpha }}^{T}}} \right)}^{H}}\text{dia}{{\text{g}}^{-1}}\left\{ \bm{\eta } \right\}\frac{\partial \bm{\chi }}{\partial {{\bm{\alpha }}^{T}}} \right\}.
\end{align}

Based on the definition of $\bm{\alpha}$ given in (\ref{alpha_parameters_unknown}),
the CRB of $f_d$ is the first diagonal element of the inverse matrix of $\mathbf{F}\left(\bm{\alpha}\right)$.
With the detailed derivations provided in Appendix \ref{Appendix_aldfikhajkldihf}, the CRB of $f_d$ is presented below.

\vspace{0.3em}

\begin{proposition}\label{proposition_aldfikhajkldihf}
The closed-form CRB of the Doppler estimation based on the CSI ratio model given in (\ref{asdlifhjaolifhaoweifh}) is:
\begin{align}\label{proposition_aldfikhajkldihf_first_formula}
\text{CR}{{\text{B}}_{{{f}_{d}}}}={{\left[ {{\mathbf{H}}^{-1}} \right]}_{1,1}},
\end{align}
where
\begin{align}\label{cal_H_result}
\mathbf{H} \! = \! \frac{2{{\left| {{\xi }_{d}} \right|}^{2}}}{\sigma _{n}^{2}}\left[ \begin{matrix}
   {{\left| {{\rho }_{2}} \right|}^{2}}\left( {{\varepsilon }_{K,2}} \! - \! {{\gamma }_{11}} \right) & \operatorname{Im}\left\{ \rho _{2}^{*}{{\varepsilon }_{K,\mu }} \right\} \! - \! {{\gamma }_{12}}  \\
   \operatorname{Im}\left\{ \rho _{2}^{*}{{\varepsilon }_{K,\mu }} \right\} \! - \! {{\gamma }_{12}} & {{\left| \frac{\partial a\left( {{\theta }_{d}} \right)}{\partial {{\theta }_{d}}} \right|}^{2}}{{\varepsilon }_{\mu ,2}} \! - \! {{\gamma }_{22}}  \\
\end{matrix} \right].
\end{align}
The intermediate variables are given in (\ref{proposition_variables_definition}) at the top of the next page.
\end{proposition}

\begin{figure*}[!t]
\begin{align}\label{proposition_variables_definition}
\nonumber  & {{\gamma }_{11}}=\frac{{{\varepsilon }_{1}}{{\left| {{\varepsilon }_{d,K,\mu }} \right|}^{2}}-2{{\varepsilon }_{K,1}}\operatorname{Re}\left\{ \varepsilon _{d,K,\mu }^{*}{{\varepsilon }_{d,\mu ,1}} \right\}+{{\varepsilon }_{\mu ,2}}\varepsilon _{K,1}^{2}}{{{\varepsilon }_{1}}{{\varepsilon }_{\mu ,2}}-{{\left| {{\varepsilon }_{d,\mu ,1}} \right|}^{2}}}, \ \
{{\gamma }_{22}}=\frac{{{\varepsilon }_{1}}{{\left| {{\varepsilon }_{d,\mu ,2}} \right|}^{2}}-2\operatorname{Re}\left\{ \varepsilon _{d,\mu ,2}^{*}{{\varepsilon }_{d,\mu ,1}}\varepsilon _{\mu ,1}^{*} \right\}+{{\varepsilon }_{\mu ,2}}{{\left| {{\varepsilon }_{\mu ,1}} \right|}^{2}}}{{{\varepsilon }_{1}}{{\varepsilon }_{\mu ,2}}-{{\left| {{\varepsilon }_{d,\mu ,1}} \right|}^{2}}}, \\
\nonumber & {{\gamma }_{12}}=\left({{{\varepsilon }_{1}}{{\varepsilon }_{\mu ,2}}-{{\left| {{\varepsilon }_{d,\mu ,1}} \right|}^{2}}}\right)^{-1}\operatorname{Im}\left\{ \rho _{2}^{*}\varepsilon _{d,\mu ,2}^{*}\left( {{\varepsilon }_{1}}{{\varepsilon }_{d,K,\mu }}-{{\varepsilon }_{K,1}}{{\varepsilon }_{d,\mu ,1}} \right)+\rho _{2}^{*}{{\varepsilon }_{\mu ,1}}\left( {{\varepsilon }_{K,1}}{{\varepsilon }_{\mu ,2}}-\varepsilon _{d,\mu ,1}^{*}{{\varepsilon }_{d,K,\mu }} \right) \right\},\\
\nonumber &\bm{\Lambda }=\text{dia}{{\text{g}}^{-1}}\left\{ {{\left| \bm{\mu } \right|}^{2}}+{{\left| a\left( {{\theta }_{d}} \right){{\rho }_{1}}\mathbf{d}\left( {{f}_{d}} \right)+{{\rho }_{0}}{{\mathbf{1}}_{K\times 1}} \right|}^{2}} \right\}\in {{\mathbb{R}}^{K\times K}}, \ \bm{\mu }={{\rho }_{1}}\mathbf{d}\left( {{f}_{d}} \right)+{{\mathbf{1}}_{K\times 1}} \in \mathbb{C}^{K\times 1}, \ {\bm{\varphi }}={{\left[ {{\varphi }_{0}},\cdots ,{{\varphi }_{K-1}} \right]}^{T}}\in {{\mathbb{Z}}^{K\times 1}},\\
\nonumber &{{\varepsilon }_{1}}=\mathbf{1}_{K\times 1}^{T}\bm{\Lambda }{{\mathbf{1}}_{K\times 1}} \in \mathbb{R},  \
{{\varepsilon }_{K,1}}=\left( 2\pi {{T}_{0}} \right)\mathbf{1}_{K\times 1}^{T}\bm{\Lambda }{\bm{\varphi }} \in \mathbb{R},\
{{\varepsilon }_{K,2}}={{\left( 2\pi {{T}_{0}} \right)}^{2}}\mathbf{1}_{K\times 1}^{T}\bm{\Lambda}{\bm{\varphi }}^{2} \in \mathbb{R}, \
{{\varepsilon }_{\mu ,2}}=\mathbf{1}_{K\times 1}^{T}\bm{\Lambda }{{\left| \bm{\mu } \right|}^{2}} \in \mathbb{R},\\
\nonumber &
{{\varepsilon }_{\mu ,1}}= \left(\frac{\partial a\left( {{\theta }_{d}} \right)}{\partial {{\theta }_{d}}}\right) \mathbf{1}_{K\times 1}^{T}\bm{\Lambda \mu } \in \mathbb{C}, \
{{\varepsilon }_{d,\mu ,2}}={{\left( \frac{\partial a\left( {{\theta }_{d}} \right)}{\partial {{\theta }_{d}}} \right)}^{*}}{{\mathbf{d}}^{H}}\left( {{f}_{d}} \right)\bm{\Lambda }{{\left| \bm{\mu } \right|}^{2}} \in \mathbb{C},\
{{\varepsilon }_{K,\mu }}=\left( 2\pi {{T}_{0}}\frac{\partial a\left( {{\theta }_{d}} \right)}{\partial {{\theta }_{d}}} \right)\mathbf{1}_{K\times 1}^{T}\bm{\Lambda }\left( {\bm{\varphi }}\odot \bm{\mu } \right) \in \mathbb{C},\\
&{{\varepsilon }_{d,\mu ,1}}={{\mathbf{d}}^{H}}\left( {{f}_{d}} \right)\bm{\Lambda \mu } \in \mathbb{C}, \
{{\varepsilon }_{d,K,\mu }}=\left( 2\pi {{T}_{0}} \right){{\mathbf{d}}^{H}}\left( {{f}_{d}} \right)\bm{\Lambda }\left( {\bm{\varphi }}\odot \bm{\mu } \right) \in \mathbb{C}.
\end{align}
\normalsize
\hrulefill
\end{figure*}

\vspace{0.3em}

It is not easy to analyze the sensing performance directly from the CRB obtained in (\ref{proposition_aldfikhajkldihf_first_formula}).
However, with proper approximations, some interesting insights can be made.
This is illustrated next.

\subsection{Approximated CRB}\label{Section_approximate_CRB}
In this subsection,
the approximated CRB is derived to facilitate further analysis of sensing performance.
The approximation here is mainly based on the value of Doppler frequency.
Intuitively, if the Doppler frequency of the dynamic target is close to 0,
the dynamic target signal will be mixed into the static object signal,
making the Doppler frequency hard to be estimated.
More rigorously, as can be verified by the results in Proposition \ref{proposition_aldfikhajkldihf},
the CRB for the Doppler frequency approaches infinity as $f_d$ becomes increasingly close to 0.
The interference of static path to the dynamic one can be evaluated by the Doppler pattern, as defined by
$\mathcal{P}\left( {\bm{\varphi }}, f_d \right)=\left|  \frac{1}{K}\sum\nolimits_{k=0}^{K-1}{{{e}^{j2\pi {{\varphi }_{k}}{{T}_{0}}{{f}_{d}}}}} \right|$.
The summation enclosed in $\left| \cdot \right|$ can be seen as the signal leakage from the static path to the moving path in the Doppler domain.

Fig. \ref{Fig_1_L2} exemplifies two cases of $\mathcal{P}\left( {\bm{\varphi }}, f_d \right)$, where $K_\text{all} = 140$.
In particular, Fig. \ref{Fig_1_L2} (a) shows the general case with $\varphi_k=k$
and $K = K_\text{all}$.
Fig. \ref{Fig_1_L2} (b) shows a specific case with $\varphi_k$ given in (\ref{Optimization_K_even}).
Note that this $\varphi_k$ is designed for enhancing sensing performance.
With details to be presented later,
we focus on illustrating the Doppler pattern here.
Despite the oscillation in Fig. \ref{Fig_1_L2} (b),
the envelope has a similar shape to the Doppler pattern shown in Fig. \ref{Fig_1_L2} (a).
One can see that the shape of $\mathcal{P}\left( {\bm{\varphi }}, f_d \right)$
or its envelope has a mainlobe around $f_d T_0 = 0$.
The width of the mainlobe basically asserts the range of Doppler values that have poor estimation performance using CSI ratios.
To depict the mainlobe region, we introduce the single-sided mainlobe width of the Doppler pattern, as given by $f_{d, \text{M}}$, which is defined as:
\begin{align}\label{mainlobe_definition}
{{f}_{d,\text{M}}}\approx \max \left\{ \arg \underset{{{f}_{d}}}{\mathop{\min }}\,\left| \mathcal{P}\left( \bm{\varphi },{{f}_{d}} \right)-0.707 \right| \right\}.
\end{align}
Note that the above definition has taken into account cases similar to Fig. \ref{Fig_1_L2} (b)
by using the envelope of the Doppler pattern.
We also note that
the above definition
does not apply to Doppler patterns with grating lobes.
This is because
through the design of $\bm{\varphi}$,
grating lobes can be readily avoided,
which will be further detailed in Section \ref{Section_IV}.

\begin{figure}[!t]
\centering
\subfigure[]
{\includegraphics[width=3.5in]{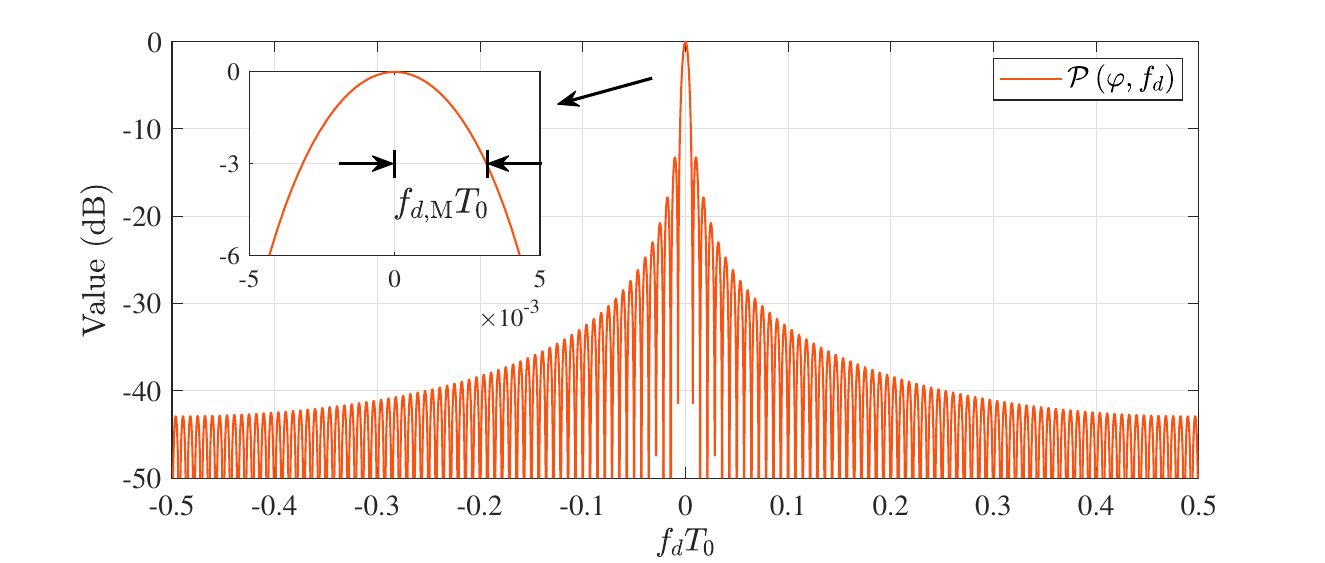}}
\subfigure[]
{\includegraphics[width=3.5in]{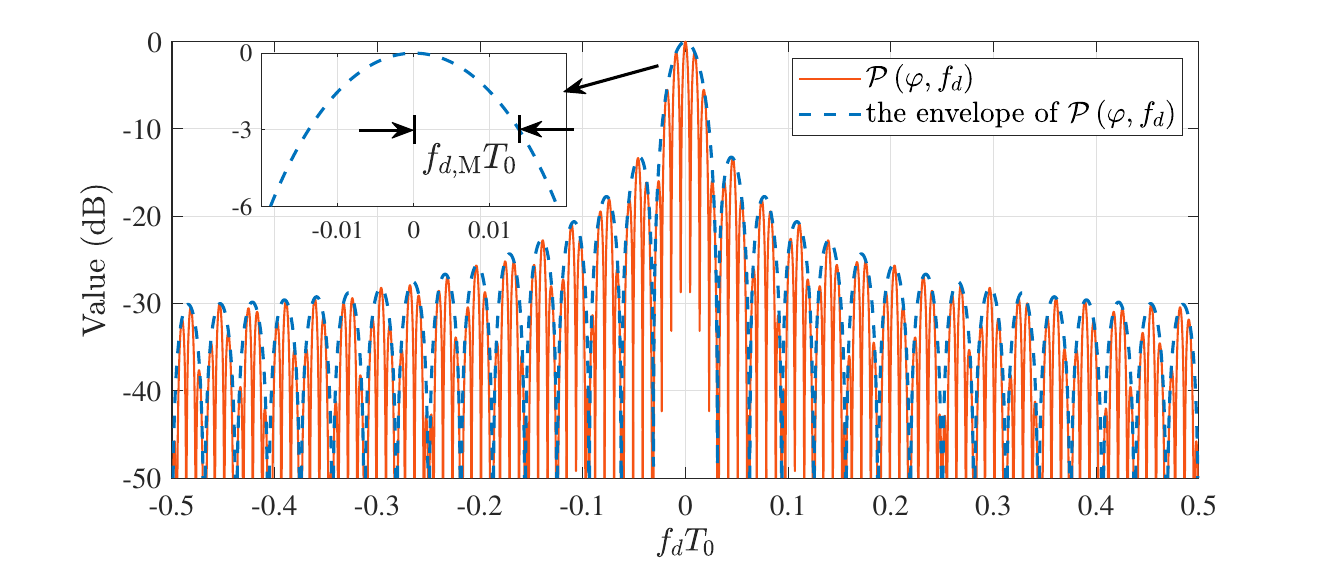}}
\caption{{The definition of the mainlobe, exemplifying two cases. (a) $K_\text{all} = K = 140$, and $\varphi_k=k$.
(b) $K_\text{all} = 140$ , $K = 64$, and $\varphi_k$ follows (\ref{Optimization_K_even}).}}
\label{Fig_1_L2}
\end{figure}

For convenience of further analysis, we also introduce $R_\text{SN}=\frac{{\left( {{\left| {{h}_{s,0}} \right|}^{2}}+{{\left| {{h}_{s,1}} \right|}^{2}} \right)}/{2}}{\sigma _{n}^{2}}$,
which
denotes the average-static-object power to noise power ratio,
$R_\text{SD}=\frac{{\left( {{\left| {{h}_{s,0}} \right|}^{2}}+{{\left| {{h}_{s,1}} \right|}^{2}} \right)}/{2}}{{{\left| {{\xi }_{d}} \right|}^{2}}}$,
which
denotes the average-static-object power to dynamic-target power ratio,
and
$R_\text{A}=\frac{{{\left| {{h}_{s,1}}-a\left( {{\theta }_{d}} \right){{h}_{s,0}} \right|}^{2}}}{{ {{\left| {{h}_{s,0}} \right|}^{2}}+{{\left| {{h}_{s,1}} \right|}^{2}} }}$.
By expanding the numerator of $R_\text{A}$,
we can obtain
${{R}_{\text{A}}}=1-\frac{2\operatorname{Re}\left\{ a\left( {{\theta }_{d}} \right){{h}_{s,0}}h_{s,1}^{*} \right\}}{{{\left| {{h}_{s,0}} \right|}^{2}}+{{\left| {{h}_{s,1}} \right|}^{2}}}$.
Thus,
$R_\text{A}$ can partially reflect the impact of the angular difference between the dynamic and static paths. Moreover,
under all possible $\theta_d$, we further have:
\begin{align}\label{formula_inequality}
0 \le \frac{{{\left( \left| {{h}_{s,0}} \right|-\left| {{h}_{s,1}} \right| \right)}^{2}}}{{{\left| {{h}_{s,0}} \right|}^{2}}+{{\left| {{h}_{s,1}} \right|}^{2}}}\le R_\text{A}\le \frac{{{\left( \left| {{h}_{s,0}} \right|+\left| {{h}_{s,1}} \right| \right)}^{2}}}{{{\left| {{h}_{s,0}} \right|}^{2}}+{{\left| {{h}_{s,1}} \right|}^{2}}} \le 2.
\end{align}
Now, we are ready to introduce the following CRB approximation.

\begin{proposition}\label{proposition_context_CRB_single_target}
When
$\left|f_d \right| > f_{d, \text{M}}$,
and
$R_\text{SD} \to 0$ or $R_\text{SD} \to +\infty$,
the Doppler CRB can be approximated by:
\begin{align}\label{CRB_closed_form_single_text}
\text{CR}{{\text{B}}_{{{f}_{d}}}}\approx \frac{1}{8{{\pi }^{2}}T_{0}^{2}K}\frac{1}{{{S}_{\varphi ,2}}-S_{\varphi ,1}^{2}}\frac{\sqrt{{{\left( 1-{{R}_{\text{SD}}} \right)}^{2}}+2{{R}_{\text{A}}}{{R}_{\text{SD}}}}}{{{R}_{\text{SN}}}{{R}_{\text{A}}}},
\end{align}
where
${{S}_{\varphi ,1}}=\frac{1}{K}\mathbf{1}_{K\times 1}^{T}{\bm{\varphi }}$;
${{S}_{\varphi ,2}}=\frac{1}{K}{\bm{\varphi }}^{T}{\bm{\varphi }}$;
and
${\bm{\varphi }}={{\left[ {{\varphi }_{0}},\cdots ,{{\varphi }_{K-1}} \right]}^{T}}\in {{\mathbb{Z}}^{K\times 1}}$.
\end{proposition}

\begin{proof}
Please refer to Appendix \ref{Appendix_proposition_context_CRB_single_target}.
Note that the relationship between $f_{d, \text{M}}$ and CRB is also established in Appendix \ref{Appendix_proposition_context_CRB_single_target}.
\end{proof}

Note that we can relax the conditions ``$R_\text{SD} \to 0$ or $R_\text{SD} \to +\infty$''
to ``$R_\text{SD} < 0.1$ or $R_\text{SD} > 8$''
based on (\ref{sdolweoidkdfn})
in Appendix \ref{Appendix_proposition_context_CRB_single_target}.
The approximated CRB has a clearer structure, compared with the original one in Proposition \ref{proposition_aldfikhajkldihf},
which can be exploited for further investigating the impact of different physical quantities on the Doppler estimation.
This is detailed next.

\subsection{Insights Drawn From the Approximated CRB}
In this subsection,
some interesting properties
of $\text{CRB}_{f_d}$
are unveiled,
highlighting the impact of different factors on the CSI ratio-based Doppler estimation.

We start with $R_\text{A}$ by presenting the following result.

\begin{corollary}\label{corollary_SSDA_what_what_aaa}
$\text{CRB}_{f_d}$ is a monotonically decreasing function of $R_\text{A}$.
\end{corollary}

\begin{proof}
The partial derivative of {$\text{CRB}_{{{f}_{d}}}$} with respect to $R_{\text{A}}$ is given by:
\begin{align}\label{formula_in_corollary_SSDAaaaa}
\frac{\partial \text{CR}{{\text{B}}_{{{f}_{d}}}}}{\partial {{R}_{\text{A}}}}\propto \frac{-{{\left( 1-{{R}_{\text{SD}}} \right)}^{2}}-{{R}_{\text{SD}}}{{R}_{\text{A}}}}{{{R}_{\text{SN}}}R_{\text{A}}^{2}\sqrt{{{\left( 1-{{R}_{\text{SD}}} \right)}^{2}}+2{{R}_{\text{A}}}{{R}_{\text{SD}}}}}.
\end{align}
Given the inequality in (\ref{formula_inequality}),
(\ref{formula_in_corollary_SSDAaaaa}) is always {smaller} than zero.
A more accurate Doppler estimation can be obtained by enhancing $R_\text{A}$.
Corollary \ref{corollary_SSDA_what_what_aaa} is thus proved.
\end{proof}

As mentioned earlier, $R_\text{A}$ is closely related with the dynamic target angle $\theta_d$,
which, hence, suggests that the CRB of Doppler estimation varies with the AoA of the dynamic path.
An interesting question follows: how sensitive the CRB can be against $\theta_d$?
To investigate this,
let us check
the
ratio of the maximum and minimum values of the CRB in $\theta_d \in \left[-\frac{\pi}{2}, \frac{\pi}{2}\right]$,
i.e.,
$\frac{{{\max }_{{{\theta }_{d}}}}{\text{CRB}}}{{{\min }_{{{\theta }_{d}}}}{\text{CRB}}}$.
A small value of $\frac{{{\max }_{{{\theta }_{d}}}}{\text{CRB}}}{{{\min }_{{{\theta }_{d}}}}{\text{CRB}}}$ illustrates that the CRB is not sensitive to the dynamic target angle.
Specifically, some interesting results are provided below.

\begin{corollary}\label{corollary_CRB_infinite}
The following three relationships hold:
\begin{subequations}
\begin{align}
\arg \underset{\sin {{\theta }_{d}}}{\mathop{\min }}\,\text{CRB}  & \approx \frac{\lambda }{2\pi d}\angle \left\{ {{{h}_{s,1}}}/{{{h}_{s,0}}} \right\}+\frac{\lambda }{d}\left( \frac{1}{2}-{{z}_{I}} \right), \label{corollary_CRB_infinite_FORMULA_a} \\
\arg \underset{\sin {{\theta }_{d}}}{\mathop{\max }}\,\text{CRB} & \approx \frac{\lambda }{2\pi d}\angle \left\{ {{{h}_{s,1}}}/{{{h}_{s,0}}} \right\}-\frac{{{z}_{A}}\lambda }{d}, \label{corollary_CRB_infinite_FORMULA_b}\\
\frac{\underset{{{\theta }_{d}}}{\mathop{\max }}\,\text{CRB}}{\underset{{{\theta }_{d}}}{\mathop{\min }}\,\text{CRB}} & \approx \frac{{{\left( \left| {{h}_{s,1}} \right|+\left| {{h}_{s,0}} \right| \right)}^{2}}}{{{\left( \left| {{h}_{s,1}} \right|-\left| {{h}_{s,0}} \right| \right)}^{2}}}. \label{corollary_CRB_infinite_FORMULA_c}
\end{align}
\end{subequations}
where
the value of ${{z}_{I}}, {{z}_{A}}\in \mathbb{Z}$
can ensure the value on the right side of
(\ref{corollary_CRB_infinite_FORMULA_a})
and
(\ref{corollary_CRB_infinite_FORMULA_b})
within $\left(-1, 1\right)$, respectively.
\end{corollary}

\begin{proof}
Using (\ref{corollary_CRB_infinite_FORMULA_b}) as an example.
Based on Corollary \ref{corollary_SSDA_what_what_aaa},
a maximal CRB can be achieved when $R_\text{A}$ is minimal,
which means $R_\text{A}=\frac{{{\left( \left| {{h}_{s,0}} \right|-\left| {{h}_{s,1}} \right| \right)}^{2}}}{{{\left| {{h}_{s,0}} \right|}^{2}}+{{\left| {{h}_{s,1}} \right|}^{2}}}$,
and we have
$\operatorname{Re}\left\{ a\left( {{\theta }_{d}} \right)\frac{{{h}_{s,0}}}{\left| {{h}_{s,0}} \right|}\frac{h_{s,1}^{*}}{\left| {{h}_{s,1}} \right|} \right\}=1$.
Due to $a\left(\theta_d\right) = {{e}^{j\frac{2\pi }{\lambda }d\sin {{\theta_d }}}}$,
the equation $\sin \theta_d=\frac{\lambda }{2\pi d}\angle \left\{ {{{h}_{s,1}}}/{{{h}_{s,0}}} \right\}$
holds
when CRB reaches its maximum.
Note that the value of $\sin \theta_d$ ranges from $-1$ to $1$.
The term $\frac{z_A \lambda }{d}$ is thus required.
The results in (\ref{corollary_CRB_infinite_FORMULA_a}) and (\ref{corollary_CRB_infinite_FORMULA_c}) can be proved similarly.
\end{proof}

From Corollary \ref{corollary_CRB_infinite}, we see that $\left|h_{s, 0}\right|$ should not
equal to ${\left| {{h}_{s,1}} \right|}$;
otherwise
the CRB approaches infinite,
suggesting that the CSI ratio cannot be used for Doppler sensing.
Moreover, we also see from Corollary \ref{corollary_CRB_infinite} that
the dynamic target angle has negligible impact on the CRB when
$ \left|{{{h}_{s,1}}}/{{{h}_{s,0}}}\right| \ll 1$ or $ \left|{{{h}_{s,1}}}/{{{h}_{s,0}}}\right| \gg 1$
due to
$\frac{{{\max }_{{{\theta }_{d}}}}\sqrt{\text{CRB}}}{{{\min }_{{{\theta }_{d}}}}\sqrt{\text{CRB}}} \approx 1$.

\vspace{0.3em}

\begin{remark}\label{remark_h_s0_equal_h_s1}
If ${\left| {{h}_{s,0}} \right|} \approx {\left| {{h}_{s,1}} \right|}$,
we can then combine all static objects in space into one {\textit{equivalent}} static object
with angle ${{\bar{\theta }}_{s}}=\arcsin \left\{ \frac{\lambda }{2\pi d}\angle \left\{ {{h}_{s,1}}/{{h}_{s,0}} \right\} \right\}$.
The CSI ratio can generate estimates with large errors when the angles of the dynamic target and the
equivalent static objects are close.
This is consistent with what is observed from Corollary \ref{corollary_CRB_infinite}.
In most cases,
this result does not occur simultaneously across all subcarriers.
We can always find subcarriers that yield relative accurate estimates.

\end{remark}

\vspace{0.3em}

Depending on whether the static path's
power is dominant, or not, compared with the dynamic path's power,
we can further simplify the CRB, as follows.

\begin{corollary}\label{corollary_ksi_large_small}
If $R_{\text{SD}} \gg 1$ (e.g., the presence of LOS path),
\begin{subequations}
\begin{align}\label{ajklnakljfn}
\nonumber \text{CR}{{\text{B}}_{{{f}_{d}}}} & \approx {{\left[ 8{{\pi }^{2}}T_{0}^{2}K\left( {{S}_{\varphi ,2}}-S_{\varphi ,1}^{2} \right) \right]}^{-1}}\frac{{{R}_{\text{SD}}}}{{{R}_{\text{SN}}}{{R}_{\text{A}}}} \\
 & ={{\left[ 8{{\pi }^{2}}T_{0}^{2}K\left( {{S}_{\varphi ,2}}-S_{\varphi ,1}^{2} \right) \right]}^{-1}}\frac{1}{{{R}_{\text{A}}}}\frac{\sigma _{n}^{2}}{{{\left| {{\xi }_{d}} \right|}^{2}}},
\end{align}
which is inversely proportional to ${{{\left| {{\xi }_{d}} \right|}^{2}}}$.
In contrast, if
$R_{\text{SD}} \ll 1$,
\begin{align}\label{formula_ksi_large}
\nonumber  &\text{CR}{{\text{B}}_{{{f}_{d}}}} \approx \left[{8{{\pi }^{2}}T_{0}^{2}K\left( {{S}_{\varphi ,2}}-S_{\varphi ,1}^{2} \right)}\right]^{-1}\frac{1}{{{R}_{\text{SN}}}{{R}_{\text{A}}}} \\
 & =\left[{8{{\pi }^{2}}T_{0}^{2}K\left( {{S}_{\varphi ,2}}-S_{\varphi ,1}^{2} \right)}\right]^{-1}\frac{2\sigma _{n}^{2}}{{{\left| {{h}_{s,1}}-a\left( {{\theta }_{d}} \right){{h}_{s,0}} \right|}^{2}}},
\end{align}
which is independent of $R_{\text{SD}}$.
\end{subequations}
\end{corollary}

\begin{proof}
Defining
$\mathcal{L}_{R_{\text{SD}}} =\sqrt{{{{\left( 1-R_\text{SD} \right)}^{2}}+2 R_\text{A} R_\text{SD}}}$
based on (\ref{CRB_closed_form_single_text}),
the Doppler CRB in (\ref{CRB_closed_form_single_text}) is inversely proportional to $\mathcal{L}_{R_{\text{SD}}}$.
We now evaluate $\mathcal{L}_{R_{\text{SD}}}$ with respect to ${R_{\text{SD}}}$.

In Case A, we have $R_\text{SD}\gg 1$,
and
${{\mathcal{L}}_{R_{\text{SD}}}}\approx R_\text{SD}$.
The term $R_{\text{SN}}$ in the numerator of (\ref{CRB_closed_form_single_text}) is combined with ${R_{\text{SD}}}$,
making the CRB inversely proportional to ${{{\left| {{\xi }_{d}} \right|}^{2}}}$.
In Case B, we have $R_\text{SD}\ll 1$,
and
${{\mathcal{L}}_{R_{\text{SD}}}}\approx 1$,
which is unrelated with $R_{\text{SD}}$.
Corollary \ref{corollary_ksi_large_small} is thus proved.
\end{proof}

The insights unveiled in this section are helpful in the waveform design for enhancing sensing performance.
This is performed next.

\section{Optimization for OFDM symbols}\label{Section_IV}
Utilizing the derived CRB and analysis made above, we proceed to optimize the OFDM symbol index ${{\varphi }_{k}}$ ($0 \le \varphi_k \le K_\text{all} - 1$, $k = 0, \cdots, K - 1$)
to improve the Doppler sensing performance using CSI ratios.
Recall that we use ${\bm{\varphi }}={{\left[ {{\varphi }_{0}},\cdots ,{{\varphi }_{K-1}} \right]}^{T}}\in {{\mathbb{Z}}^{K\times 1}}$
to collect indexes of $K$ out of $K_\text{all}$ OFDM symbols used for sensing.

\subsection{Strategy of Optimizing $\bm{\varphi}$}\label{Section_strategy_of_optimizing}

Based on Section \ref{Section_approximate_CRB},
a low CRB can be achieved by optimizing two terms,
which are
$\mathcal{L}_1\left( {\bm{\varphi }} \right)={{S}_{\varphi ,2}}-S_{\varphi ,1}^{2}$ (unrelated with $f_d$)
and
{the Doppler pattern}
$\mathcal{P}\left( {\bm{\varphi }}, f_d \right)=\left|  \frac{1}{K}\sum\nolimits_{k=0}^{K-1}{{{e}^{j2\pi {{\varphi }_{k}}{{T}_{0}}{{f}_{d}}}}} \right|$
(related with $f_d$ {and defined in Section \ref{Section_approximate_CRB}}).
These further lead to the following metrics.
\begin{itemize}
\item
Value of $\mathcal{L}_1\left( {\bm{\varphi }} \right)$:
The value of $\mathcal{L}_1\left( {\bm{\varphi }} \right)$ has significant impact on CRB.
As indicated by (\ref{CRB_closed_form_single_text}), $\text{CRB}_{{{f}_{d}}}$ is inversely proportional to ${{\mathcal{L}}_{1}}\left( \bm{\varphi } \right)$.
The term ${{\varphi }_{k}}$ should be selected by maximizing
$\mathcal{L}_1\left( {\bm{\varphi }} \right)$;

\item
Mainlobe width
of {$\mathcal{P}\left( {\bm{\varphi }}, f_d \right)$}:
The mainlobe width
of {$\mathcal{P}\left( {\bm{\varphi }}, f_d \right)$}
(i.e., {$f_{d, \text{M}}$} defined in {Section \ref{Section_approximate_CRB}})
should be as narrow as possible.
As aforementioned in Section \ref{Section_approximate_CRB},
the CRB will sharply increase when $f_d$ is close to the mainlobe of CRB.
In some scenarios, such as indoor localization and vital sign detection,
where the Doppler of the dynamic target is usually low,
a narrow mainlobe (i.e., a small {$f_{d, \text{M}}$}) is able to improve the estimation accuracy;

\item
Sidelobe of {$\mathcal{P}\left( {\bm{\varphi }}, f_d \right)$}:
The sidelobe of {$\mathcal{P}\left( {\bm{\varphi }}, f_d \right)$}
causes
the fluctuation of the CRB under different $f_d$ {values},
which
can render large
CRB values at some Doppler frequencies.
Thus, reducing sidelobes and ensuring uniformness in a Doppler frequency region of interest can help ensure a high and uniform sensing performance.
\end{itemize}

Based on the above analysis,
our strategy of optimizing $\bm{\varphi}$ is given as:
\emph{the optimization should
maximize the value of $\mathcal{L}_1\left( {\bm{\varphi }} \right)$ and
minimize the mainlobe width ({$f_{d, \text{M}}$})
while keeping the sidelobe of {$\mathcal{P}\left( {\bm{\varphi }}, f_d \right)$}
minimal.}

Note that these performance metrics play critical roles in different sensing scenarios.
In particular,
when static paths are strong and the Doppler frequency of the dynamic path is small,
the dynamic path is more likely to be severely interfered by the static paths.
This scenario is interference-limited,
where more focus can be placed on reducing the mainlobe width of $\mathcal{P}\left( {\bm{\varphi }}, f_d \right)$.
In contrast,
when the Doppler frequency of the dynamic path is far from zero with little interference from static paths,
its estimation performance is mainly determined by noises,
i.e., a noise-limited case.
In such a case,
CRB intrinsically asserts the estimation performance,
and hence $\mathcal{L}_1\left( {\bm{\varphi }} \right)$ can be mainly adopted for waveform design.
Next, we illustrate the CRB-driven waveform designs in these two cases.

\subsection{Noise-limited Waveform Design}\label{Section_generlaaa}

Consider noise-limited cases first.
We can then focus on $\mathcal{L}_1\left( {\bm{\varphi }} \right)$ only,
as $\mathcal{P}\left( {\bm{\varphi }}, f_d \right)$ mainly asserts the mutual interference between static and dynamic paths.

In such cases,
the optimal OFDM symbol index can be obtained by maximizing $\mathcal{L}_1\left( {\bm{\varphi }} \right)$.
With the CRB results in Proposition \ref{proposition_context_CRB_single_target},
one can validate the following result.
The details are suppressed here for conciseness.

\vspace{0.3em}

\begin{proposition}\label{proposition_general_optimazition_result}
(Noise-limited Doppler Sensing)
For $0 \le \varphi_k \le K_\text{all} - 1$ and $k = 0, \cdots, K - 1$,
the OFDM symbol ${{\varphi }_{k}}$
{indexes collected by the following vector maximizes $\mathcal{L}_1\left( {\bm{\varphi }} \right)$:}
\begin{align}\label{Optimization_K_even}
  & \bm{\varphi }_{\text{op}}^{\text{gen}}={{\left[{{\bm\kappa }}_{{K}/{2}, \text{A}}^{T},{\left({{K}_{\text{all}}}-1\right)\mathbf{1}_{{{K}/{2}}\times 1}^{T}}-{{\bm\kappa }}_{{K}/{2}, \text{B}}^{T} \right]}^{T}},
\end{align}
where
${{\bm\kappa }}_{{K}/{2}, \text{A}} = {{\bm\kappa }}_{{K}/{2}, \text{B}} = {{\left[ 0,1,\cdots ,\left\lfloor K/2 \right\rfloor -1 \right]}^{T}}\in {{\mathbb{Z}}^{\left\lfloor K/2 \right\rfloor \times 1}}$
if $K$ is even.
If $K$ is odd,
${{\bm\kappa }}_{{K}/{2}, \text{A}} = {{\left[ 0,1,\cdots ,\left\lceil K/2 \right\rceil -1 \right]}^{T}}\in {{\mathbb{Z}}^{\left\lceil K/2 \right\rceil \times 1}}$ and
${{\bm\kappa }}_{{K}/{2}, \text{B}} = {{\left[ 0,1,\cdots ,\left\lfloor K/2 \right\rfloor -1 \right]}^{T}}$,
or
${{\bm\kappa }}_{{K}/{2}, \text{A}} = {{\left[ 0,1,\cdots ,\left\lfloor K/2 \right\rfloor -1 \right]}^{T}}$
and
${{\bm\kappa }}_{{K}/{2}, \text{B}} = {{\left[ 0,1,\cdots ,\left\lceil K/2 \right\rceil -1 \right]}^{T}}$.
\end{proposition}

\vspace{0.3em}

Proposition \ref{proposition_general_optimazition_result} suggests that,
to optimize the theoretical Doppler sensing performance using CSI ratios in noise-limited cases,
the OFDM symbols should be selected \emph{{near the two ends} of the available OFDM symbols}.
Moreover, the symbol indexes are symmetrical to the center and
occupy the maximum observation time of the available symbols.
Additionally, the symmetry pattern in $\bm{\varphi }_{\text{op}}^{\text{gen}}$
results in a nice feature of the Doppler pattern.
The feature can be used for further simplifying the waveform design;
please refer to Appendix \ref{Appendix_of_Lemma_Envelope_in_aoifhnaol} for its proof.

\vspace{0.3em}

\begin{proposition}\label{Lemma_Envelope_in_aoifhnaol}
If the indexes of the OFDM symbols used for sensing are symmetric against the point ${\left(K_\text{all} - 1\right)}/{2}$,
i.e., $\bm{\varphi }$ has a structure of $\bm{\varphi }={{\left[ {{{\bm{\tilde{\varphi }}}}^{T}},{{K}_{\text{all}}}-1-{{{\bm{\tilde{\varphi }}}}^{T}} \right]}^{T}}$,
where
$\bm{\tilde{\varphi }}\in {{\mathbb{Z}}^{{K}/{2}\times 1}}$,
the envelope of $\mathcal{P}\left( {\bm{\varphi }}, f_d \right)$
is solely determined by half indexes and can be given by:
\begin{align}\label{afoliakjlafihjdofih}
{{\mathcal{E}_\mathcal{P}}}\left( \bm{\tilde{\varphi }},{{f}_{d}} \right)=\left| \frac{2}{K}\sum\limits_{k=0}^{\frac{K}{2}-1}{{{e}^{j2\pi {{\left[ {\bm{\tilde{\varphi }}} \right]}_{k}}{{T}_{0}}{{f}_{d}}}}} \right|.
\end{align}
\end{proposition}

\vspace{0.3em}

Employing Proposition \ref{Lemma_Envelope_in_aoifhnaol},
the envelope of
{$\mathcal{P}\left( {{\bm{\varphi }}_{\text{op}}^{\text{gen}}}, f_d \right)$}
is
${\mathcal{E}_\mathcal{P}}\left( {{\bm{\tilde{\varphi} }}_{\text{op}}^{\text{gen}}}, f_d \right)=\left| \frac{\text{sinc}\left(  \frac{K}{2}{{T}_{0}}{{f}_{d}} \right)}{\text{sinc}\left(  {{T}_{0}}{{f}_{d}} \right)} \right|$
whose mainlobe width is related with ${K}/{2}$ instead of $K_\text{all}$.
{In Fig. \ref{Fig_1_L2} (b),
we plot {$\mathcal{P}\left( {{\bm{\varphi }}_{\text{op}}^{\text{gen}}}, f_d \right)$} and its envelope
${\mathcal{E}_\mathcal{P}}\left( {{\bm{\tilde{\varphi} }}_{\text{op}}^{\text{gen}}}, f_d \right)$
based on Proposition \ref{proposition_general_optimazition_result} and Proposition \ref{Lemma_Envelope_in_aoifhnaol}.}
It can be seen that
$\mathcal{P}\left( {{\bm{\varphi }}_{\text{op}}^{\text{gen}}}, f_d \right)$
is modulated by a ``high-frequency oscillating signal'', which is essentially caused by the signal in (\ref{alidkfjaldikhjflahf}).
In such cases, the envelope
${\mathcal{E}_\mathcal{P}}\left( {{\bm{\tilde{\varphi} }}_{\text{op}}^{\text{gen}}}, f_d \right)$
plays a critical role in the Doppler sensing performance rather than $\mathcal{P}\left( {{\bm{\varphi }}_{\text{op}}^{\text{gen}}}, f_d \right)$.

As illustrated in Section \ref{Section_III},
the mainlobe of the Doppler pattern or its envelope asserts the strong interference between static and dynamic paths.
A wide mainlobe or high sidelobes can degrade Doppler sensing performance.
Next, we show that designing the sensing symbol indexes in $\bm{\varphi}$ can help minimize mainlobe width and sidelobe levels in a Doppler pattern.

\subsection{Interference-limited Waveform Design}\label{Section_IV_weighted_balabala}
In this subsection,
we optimize ${{\varphi }_{k}}$
{to reduce}
the mainlobe width and the sidelobe level
of
${\mathcal{E}_\mathcal{P}}\left( {\bm{\tilde\varphi }}, f_d \right)$
{for better sensing performance}.
For ease of illustration,
we {focus on the case of even} $K$.

Proposition \ref{proposition_general_optimazition_result} indicates that the maximal $\mathcal{L}_1\left( {\bm{\varphi }} \right)$ is obtained
when $\varphi_k$ is symmetrically located near the two ends of all available symbols.
Now,
we need to release this constraint
in exchange of the reduced mainlobe and sidelobe of Doppler pattern.
We introduce the structure of $\bm{\varphi}$,
as depicted in Fig. \ref{Fig_text_optimization_figure},
to factor in the trade-off between optimizing $\mathcal{L}_1\left( {\bm{\varphi }} \right)$
and
$\mathcal{P}\left( {\bm{\varphi }}, f_d \right)$.
Specifically,
we fix $2K_F$ OFDM symbols
near the two ends of $\bm{\varphi }_{\text{op}}^{\text{gen}}$
as given in (\ref{Optimization_K_even}) and allow the remaining $2K_B$ OFDM symbols
to be reconfigured for optimizing $\mathcal{P}\left( {\bm{\varphi }}, f_d \right)$.
Since $K_F + K_B = {K}/{2}$,
the value of $K_B$ controls the trade off between the two objectives.
In an extreme case of $K_B={{K}/{2}}$, the complete focus is placed on reducing the interference between static and dynamic paths.

Based on Fig. \ref{Fig_text_optimization_figure},
${\bm{\tilde{\varphi }}}$ can be further rewritten as
$\bm{\tilde{\varphi }}={{\left[ \bm{\tilde{\varphi }}_{F}^{T},\bm{\tilde{\varphi }}_{B}^{T} \right]}^{T}}\in {{\mathbb{Z}}^{{K}/{2}\;\times 1}}$,
where
${{{\bm{\tilde{\varphi }}}}_{F}}={{\left[ 0,1,\cdots ,{{K}_{F}}-1 \right]}^{T}}\in {{\mathbb{Z}}^{{{K}_{F}}\times 1}}$,
${{{\bm{\tilde{\varphi }}}}_{B}}\in {{\mathbb{Z}}^{{{K}_{B}}\times 1}}$.
{This leaves} ${{{\bm{\tilde{\varphi }}}}_{B}}$
{for optimization next}.

Since ${{{\bm{\tilde{\varphi }}}}_{B}}$
is an integer vector,
it can be solved via an integer optimization problem,
which is generally challenging to solve.
Fortunately,
we have the following result to substantially simply the task of optimizing $\bm{\tilde{\varphi}}_B$.

\begin{figure}[!t]
\centering
\includegraphics[width=3.4in]{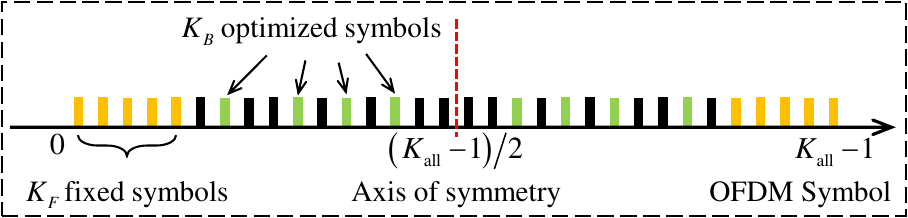}
\caption{Illustration of the structure of $\bm{\varphi}$ in optimization algorithm,
with $K_\text{all} = 32$, $K = 18$, $K_F = 5$, and $K_B = 4$,
where $K_F + K_B = {K}/{2}$.}
\label{Fig_text_optimization_figure}
\end{figure}

\begin{proposition}\label{proposition_integer_operation}
For ${\bm{{\tilde{\varphi }}}_{B}^{\prime}} = {\bm{{\tilde{\varphi }}}_{B}} + \bm{\delta}_B$,
where each entry of $\bm{\delta}_B$ is continuous in $\left(-0.5, 0.5\right)$, and hence,
${\bm{{\tilde{\varphi }}}_{B}} = \text{round}\left\{{\bm{{\tilde{\varphi }}}_{B}^{\prime}}\right\}$,
we have:
\begin{align}\label{round_operation_relationship}
\left| {\mathcal{E}_\mathcal{P}}\left( \left[ \begin{matrix}
   {{{\bm{\tilde{\varphi }}}}_{F}}  \\
   {{{\bm{\tilde{\varphi }}}}_{B}}  \\
\end{matrix} \right],{{f}_{d}} \right)-{\mathcal{E}_\mathcal{P}}\left( \left[ \begin{matrix}
   {{{\bm{\tilde{\varphi }}}}_{F}}  \\
   {{{\bm{\tilde{\varphi }}}}_{B}^{\prime}}  \\
\end{matrix} \right],{{f}_{d}} \right) \right| < {\frac{{\pi{K}_{B}}}{{K}/{2}}{\left|{{T}_{0}}{{{f}_{d}}}\right|}},
\end{align}
where
$\left|{{T}_{0}}{{f}_{d}}\right| \le {{T}_{0}}{{f}_{d, \text{M}}}$.
\end{proposition}

\begin{proof}
Please refer to Appendix \ref{Appendix_of_Proposition_Integer_Operation}.
\end{proof}

We note that the right-hand side of (\ref{round_operation_relationship}) is generally smaller than ${{\pi}/{30}}$.
As seen from Fig. \ref{Fig_1_L2}, $\left| {T_0}{f_d} \right| \le {T_0} {f_{d,\text{M}}} < {{1}/{30}}$.
Note that Fig. \ref{Fig_1_L2} (b) actually illustrates the maximum mainlobe width.
The vector $\bm{\varphi}$ used for plotting the figure is obtained by maximizing $\mathcal{L}_1\left( {\bm{\varphi }} \right)$ in Proposition \ref{proposition_general_optimazition_result} without considering the Doppler pattern.
To further validate this, let us check a numerical example.
For instance, considering $\bm{\varphi }_{\text{op}}^{\text{gen}}$ defined in (\ref{Optimization_K_even}),
its mainlobe width can be calculated as ${{T}_{0}}{{f}_{d,\text{M}}}\approx \frac{0.44}{{K}/{2}}$.
When $K = 128$ and $K_B = 50$, (\ref{round_operation_relationship}) is lower than $1.69\times {{10}^{-2}}$.

Proposition \ref{proposition_integer_operation} suggests that
the rounding operation has negligible impact on
${\mathcal{E}_\mathcal{P}}\left( {\bm{\tilde\varphi }}, f_d \right)$'s mainlobe.
This enables us to conveniently solve a non-integer optimization problem
for optimizing continuous
${\bm{{\tilde{\varphi }}}_{B}^{\prime}}$
and then rounding it element-wise to obtain
${\bm{{\tilde{\varphi }}}_{B}}$.

It is obvious that simultaneously minimizing the mainlobe width and sidelobe level of patterns shown in Fig. \ref{Fig_1_L2} are contradicting objectives.
Thus, we choose to minimize the sidelobe level, subject to an upper bound of the mainlobe width, as denoted by
$f_{d, \text{M}}^{\text{ex}}$.
For the sidelobe suppression, we adopt the
weighted sidelobe power ${{\mathcal{P}}_{\text{SL}}}$,
as given by:
\begin{align}\label{weighted_sidelobe_power_aasdlka}
{{\mathcal{P}}_{\text{SL}}}=\int_{f_{d,\text{M}}^{\text{ex}}}^{\frac{1}{2{{T}_{0}}}}{\varpi \left( {{f}_{d}} \right){{\left| {\mathcal{E}_\mathcal{P}}\left( {{{\bm{\tilde{\varphi }}}}},{{f}_{d}} \right) \right|}^{2}}d{{f}_{d}}},
\end{align}
where $\varpi \left( {{f}_{d}} \right)$
denotes the weighted function with respect to $f_d$.
Note that
$\varpi \left( {{f}_{d}} \right)$
controls the proportion of
${\mathcal{E}_\mathcal{P}}\left( {\bm{\tilde\varphi }}, f_d \right)$ on different $f_d$ in ${{\mathcal{P}}_{\text{SL}}}$.
We can formulate the following optimization problem to achieve the strategy illustrated in Section \ref{Section_strategy_of_optimizing}:
\begin{align}\label{optimal_formula_in_phi_optimization}
\nonumber \underset{\begin{smallmatrix}
 {{{\bm{\tilde{\varphi }}}}_{B}} \\
\end{smallmatrix}}{\mathop{\min }} & \ \ {{\mathcal{P}}_{\text{SL}}} \\
\nonumber \text{s.t.} \ \, & {{\mathbf{1}}_{{{K}_{B}}\times 1}}{{K}_{F}}\le {{{\bm{\tilde{\varphi }}}}_{B}}\le {{\mathbf{1}}_{{{K}_{B}}\times 1}}\left\lfloor \frac{{{K}_{\text{all}}}-1}{2} \right\rfloor \\
\nonumber & f_{d,\text{M}}^{\text{ex}}\le {\left|{{f}_{d}}\right|}\le \frac{1}{2{{T}_{0}}}\\
\nonumber & \bm{\varphi }={{\left[ {{{\bm{\tilde{\varphi }}}}^{T}},{{K}_{\text{all}}}-1-{{{\bm{\tilde{\varphi }}}}^{T}} \right]}^{T}} \\
 & \bm{\tilde{\varphi }}={{\left[ \bm{\tilde{\varphi }}_{F}^{T},\bm{\tilde{\varphi }}_{B}^{T} \right]}^{T}}.
\end{align}

Inspired by \cite{8391730},
we introduce a perturbation on $\bm{\tilde{\varphi}_B}$
and iteratively {optimize} the perturbation
{to minimize the objective function in (\ref{weighted_sidelobe_power_aasdlka}).}
The vector updated in the ${\left( i+1 \right)}$ step is expressed as:
\begin{align}\label{update_vector_phi_B}
\bm{\tilde{\varphi }}_{B}^{\left( i+1 \right)}=\bm{\tilde{\varphi }}_{B}^{\left( i \right)}+\Delta \bm{\tilde{\varphi }}_{B}^{\left( i+1 \right)},
\end{align}
where
$\Delta \bm{\tilde{\varphi }}_{B}^{\left( i+1 \right)} \in \mathbb{R}^{K_B \times 1}$
denotes
a perturbation amount for $\bm{\tilde{\varphi }}_{B}^{\left( i \right)}$.
Note that the value range of each element in $\bm{\tilde{\varphi }}_{B}$ is from
${{K}_{F}}$
to
$\left\lceil \frac{{{K}_{\text{all}}}-1}{2} \right\rceil -1$.
We consider the elements of
the initialization vector $\bm{\tilde{\varphi }}_{B}^{\left( 0 \right)}$
distributed at equal intervals in the above range.
This ensures that the sidelobes nearby the mainlobe have relatively small values.
The $i$-th element of $\bm{\tilde{\varphi }}_{B}^{\left( 0 \right)}$ is thus given by:
\begin{align}\label{asdlfjihasdkolguhweuihf}
{{\left[ \bm{\tilde{\varphi }}_{B}^{\left( 0 \right)} \right]}_{i}}={{K}_{F}}+\text{round}\left\{ i\frac{\left\lceil \frac{{{K}_{\text{all}}}-1}{2} \right\rceil -1-{{K}_{F}}}{{{K}_{B}}-1} \right\},
\end{align}
where
$i=0,\cdots ,{{K}_{B}}-1$.

Assume that the value of $\Delta \bm{\tilde{\varphi }}_{B}^{\left( i+1 \right)} $ is small,
${{\mathcal{E}_\mathcal{P}}}\left( \bm{\tilde{\varphi }},{{f}_{d}} \right)$,
as given in (\ref{afoliakjlafihjdofih}), can be approximated by:
\begin{align}\label{asdohfiljhaljh}
\nonumber  & {\mathcal{E}_\mathcal{P}}\left( \bm{\tilde{\varphi }},{{f}_{d}} \right)\\
\nonumber   &=\frac{2}{K}\left| \sum\limits_{k=0}^{{{K}_{F}}-1}{{{e}^{j2\pi {{\left[ {{{\bm{\tilde{\varphi }}}}_{F}} \right]}_{k}}{{T}_{0}}{{f}_{d}}}}}+\sum\limits_{k=0}^{{{K}_{B}}-1}{{{e}^{j2\pi {{\left[ \bm{\tilde{\varphi }}_{B}^{\left( i \right)}+\Delta \bm{\tilde{\varphi }}_{B}^{\left( i+1 \right)} \right]}_{k}}{{T}_{0}}{{f}_{d}}}}} \right| \\
\nonumber  & =\frac{2}{K}\left| {{\mathbf{f}}^{T}}\left( {{{\bm{\tilde{\varphi }}}}_{F}},{{f}_{d}} \right){{\mathbf{1}}_{{{K}_{F}}\times 1}}+{{\mathbf{f}}^{T}}\left( \bm{\tilde{\varphi }}_{B}^{\left( i \right)},{{f}_{d}} \right)\mathbf{f}\left( \Delta \bm{\tilde{\varphi }}_{B}^{\left( i+1 \right)},{{f}_{d}} \right) \right| \\
 & \approx \frac{2}{K}\left| {{\mathbf{f}}^{T}}\left( \bm{\tilde{\varphi }},{{f}_{d}} \right){{\mathbf{1}}_{{K}/{2}\;\times 1}}+j2\pi {{T}_{0}}{{f}_{d}}{{\mathbf{f}}^{T}}\left( \bm{\tilde{\varphi }}_{B}^{\left( i \right)},{{f}_{d}} \right)\Delta \bm{\tilde{\varphi }}_{B}^{\left( i+1 \right)} \right|,
\end{align}
where
$\mathbf{f}\left( \bm{\tilde{\varphi }},{{f}_{d}} \right)\triangleq {{e}^{j2\pi \bm{\tilde{\varphi }}{{T}_{0}}{{f}_{d}}}}$.

This enables
the optimization problem
(\ref{optimal_formula_in_phi_optimization}) to be
converted
into
minimizing the following function in terms of $\Delta \bm{\tilde{\varphi }}_{B}$:
\begin{align}\label{asukoldj4398iokdf}
\mathcal{H}_{\text{op}}^{\left( i+1 \right)}=\mathcal{P}_{\text{SL}}^{\left( i+1 \right)}+a{{\left( {{\mathbf{w}}^{\left( i \right)}} \right)}^{T}}{{\left( \Delta \bm{\tilde{\varphi }}_{B}^{\left( i+1 \right)} \right)}^{2}},
\end{align}
where
$a \in \mathbb{R}^{+}$ is a weight, which controls the convergence rate.
Moreover,
${{\mathbf{w}}^{\left( i \right)}}={{\left[ w_{0}^{\left( i \right)},\cdots ,w_{{{K}_{B}}-1}^{\left( i \right)} \right]}^{T}}\in {{\mathbb{R}}^{{{K}_{B}}\times 1}}$
is a weighting vector restricting the value range of
${\bm{{\tilde{\varphi }}}_{B}^{\left(i\right)}}$, where
${{w}_{{{k}_{B}}}^{\left(i\right)}}=1$
if
${{K}_{F}}+1<{{\left[ {{{\bm{\tilde{\varphi }}}}_{B}^{\left(i\right)}} \right]}_{{{k}_{B}}}}<\left\lfloor \frac{{{K}_{\text{all}}}-1}{2} \right\rfloor -1$,
and
${{w}_{{{k}_{B}}}^{\left(i\right)}}=\inf$
if
${{\left[ {{{\bm{\tilde{\varphi }}}}_{B}^{\left(i\right)}} \right]}_{{{k}_{B}}}}\le {{K}_{F}}+1$
or
${{\left[ {{{\bm{\tilde{\varphi }}}}_{B}^{\left(i\right)}} \right]}_{{{k}_{B}}}}\ge \left\lfloor \frac{{{K}_{\text{all}}}-1}{2} \right\rfloor -1$.
Note that the second term in (\ref{asukoldj4398iokdf})
{is subject to the fact that} each element in $\Delta \bm{\tilde{\varphi }}_{B}^{\left( i+1 \right)} $ has a small value.

As derived in Appendix \ref{Appendix_iteration_algorithm_result}, the
updated vector $\Delta \bm{\tilde{\varphi }}_{B}^{\left( i+1 \right)}$ can be solved as:
\begin{align}\label{iteration_algorithm_result}
 \Delta \bm{\tilde{\varphi }}_{B}^{\left( i+1 \right)}={{\left( \operatorname{Re}\left\{ \mathbf{V}_{2}^{\left( i \right)} \right\}+a\text{diag}\left\{ \mathbf{w}^{\left( i \right)} \right\} \right)}^{-1}}\operatorname{Im}\left\{ \mathbf{V}_{1}^{\left( i \right)} \right\}{{\mathbf{1}}_{\frac{K}{2}\times 1}},
\end{align}
where
$\mathbf{V}_{1}^{\left( i \right)}=\frac{8\pi {{T}_{0}} }{{{K}^{2}}}\int_{f_{d,\text{M}}^{\text{ex}}}^{\frac{1}{2{{T}_{0}}}}{\varpi \left( {{f}_{d}} \right){{f}_{d}}\mathbf{f}\left( \bm{\tilde{\varphi }}_{B}^{\left( i \right)},{{f}_{d}} \right){{\mathbf{f}}^{H}}\left( \bm{\tilde{\varphi }},{{f}_{d}} \right)d{{f}_{d}}}$,
$\mathbf{V}_{2}^{\left( i \right)}=\frac{4{{\left( 2\pi {{T}_{0}} \right)}^{2}}}{{{K}^{2}}}\int_{f_{d,\text{M}}^{\text{ex}}}^{\frac{1}{2{{T}_{0}}}}{\varpi \left( {{f}_{d}} \right)f_{d}^{2}\mathbf{f}\left( \bm{\tilde{\varphi }}_{B}^{\left( i \right)},{{f}_{d}} \right){{\mathbf{f}}^{H}}\left( \bm{\tilde{\varphi }}_{B}^{\left( i \right)},{{f}_{d}} \right)d{{f}_{d}}}$,
and
$\mathbf{V}_{1}^{\left( i \right)}\in {{\mathbb{C}}^{{{K}_{B}}\times \frac{K}{2}}}$,
$\mathbf{V}_{2}^{\left( i \right)}\in {{\mathbb{C}}^{{{K}_{B}}\times {{K}_{B}}}}$.
When $\left| \Delta \bm{\tilde{\varphi }}_{B}^{\left( i+1 \right)} - \Delta \bm{\tilde{\varphi }}_{B}^{\left( i \right)} \right|$
is lower than a preset threshold, the iteration converges and
we obtain $ \bm{\tilde{\varphi }}_{B}^{\left(I\right)}$
and then take the rounding of $ \bm{\tilde{\varphi }}_{B}^{\left(I\right)}$
as the optimized $ \bm{\tilde{\varphi }}_{B}$:
\begin{align}\label{erowklujfhnedljkvgn}
\bm{\tilde{\varphi }}_{B}^{\text{op}} = \text{round}\left\{\bm{\tilde{\varphi }}_{B}^{\left(I\right)}\right\}.
\end{align}
Finally, we construct the optimal $\bm{\varphi}$:
\begin{align}\label{34oikldfdf}
{{\bm{\varphi }}^{\text{op}}} \! = \! {{\left[ \bm{\tilde{\varphi }}_{F}^{T},{{\left( \bm{\tilde{\varphi }}_{B}^{\text{op}} \right)}^{T}},{{K}_{\text{all}}} \! - \! 1 \! - \! \bm{\tilde{\varphi }}_{F}^{T},{{K}_{\text{all}}} \! - \! 1 \! - \! {{\left( \bm{\tilde{\varphi }}_{B}^{\text{op}} \right)}^{T}} \right]}^{T}}.
\end{align}

Algorithm \ref{Algorithm1}
summarizes the proposed waveform design for the interference-limited Doppler sensing using CSI ratios.
\begin{algorithm}[!t]
\caption{Proposed Waveform Design for Interference-Limited Doppler Sensing}\label{Algorithm1}
\label{Algorithm1}
\hspace*{0.02in}{\bf Input:}
the number of the available OFDM symbols $K_\text{all}$,
the number of the symbols allocated for sensing $K$,
the number of the fixed symbols $K_F$,
the expected mainlobe width $f_{d, \text{M}}^{\text{ex}}$,
the symbol interval $T_0$,
and the weighting function $\varpi \left( {{f}_{d}} \right)$.\\
\hspace*{0.02in}{\bf Output:}
the optimal OFDM symbol vector ${{\bm{\varphi }}^{\text{op}}}$.

\begin{algorithmic}[1]
\STATE
Constructing ${{{\bm{\tilde{\varphi }}}}_{F}}={{\left[ 0,1,\cdots ,{{K}_{F}}-1 \right]}^{T}}\in {{\mathbb{Z}}^{{{K}_{F}}\times 1}}$
and
calculating $K_B = {K}/{2} - K_F$.\\
\STATE
\textbf{Initialize:}
Performing (\ref{asdlfjihasdkolguhweuihf}) and $i=0$.
\WHILE { $\Delta \bm{\tilde{\varphi }}_{B}^{\left( i+1 \right)} - \Delta \bm{\tilde{\varphi }}_{B}^{\left( i \right)}>$ threshold}
\STATE
Updating ${{\mathbf{w}}^{\left( i \right)}}$ by
${\bm{{\tilde{\varphi }}}_{B}^{\left(i\right)}}$
and calculating
$\mathbf{V}_{1}^{\left( i \right)}$
and
$\mathbf{V}_{2}^{\left( i \right)}$.
\STATE
Calculate
$\Delta \bm{\tilde{\varphi }}_{B}^{\left( i+1 \right)}$
by (\ref{iteration_algorithm_result})
and updating
$\bm{\tilde{\varphi }}_{B}^{\left( i+1 \right)}$
by
(\ref{update_vector_phi_B}).
\ENDWHILE
\STATE
Performing (\ref{erowklujfhnedljkvgn}) and (\ref{34oikldfdf}) for ${{\bm{\varphi }}^{\text{op}}}$.
\end{algorithmic}
\end{algorithm}

\section{Simulation Results}\label{Section_simulation_results}
In this section, simulation results are presented to
validate the preceding CRB derivations, analysis and waveform design.
The carrier frequency is set to $3\text{GHz}$, so the wavelength is $\lambda = 0.1\text{m}$.
Suppose that the CSI-ratio model utilizes two received antennas with the spacing $d = {{\lambda}/{2}} = 0.05 \text{m}$.
The time interval between OFDM symbols is $T_0 = 125\mu s$,
so the range of unambiguous Doppler is
$\left| {{f}_{d}} \right|\le \frac{1}{2{{T}_{0}}}=4000\text{Hz}$,
i.e.,
$\left| v \right|\le \frac{\lambda }{4{{T}_{0}}}=200{\text{m}}/{\text{s}}$.
Unless otherwise specified,
we set the number of the allocated OFDM symbols used for uplink sensing to $K = 128$
and
the set the available number of OFDM symbols to $K_\text{all } = 512$.
The OFDM symbol index follows the noise-limited Doppler sensing method given in (\ref{Optimization_K_even}).

The radial velocity of the dynamic target is set to $5{{\text{m}}/{\text{s}}}$, so the Doppler frequency is $f_d = 100\text{Hz}$,
and the angle of the dynamic target is $\theta_d = 10^\circ$.
We set
${{\rho }_{0}}=\frac{{{h}_{s,1}}}{{{h}_{s,0}}} = 1.2 e^{-j{\frac{30\pi}{180}}}$,
and
${{\rho }_{1 }}=\frac{{{\xi }_{d }}}{{{h}_{s,0}}} = 0.1 e^{-j{\frac{110\pi}{180}}}$.
Therefore,
we have
$R_\text{SD} = 122$ (i.e., 20.86dB)
and
$R_\text{A} = 0.527$.
To meet the high-SNR requirement given in Lemma \ref{fundamental_assumption_high_SNR},
a small value of the noise power is required
and
we set
$R_\text{SN} = 1220$ (i.e., 30.86dB).

\subsection{Evaluation of the CRB}\label{Section_simulation_BBBB}
In this subsection, the simulation and theoretical results are provided for
evaluating the CRB.

\begin{figure}[!t]
\centering
\includegraphics[width=3.2in]{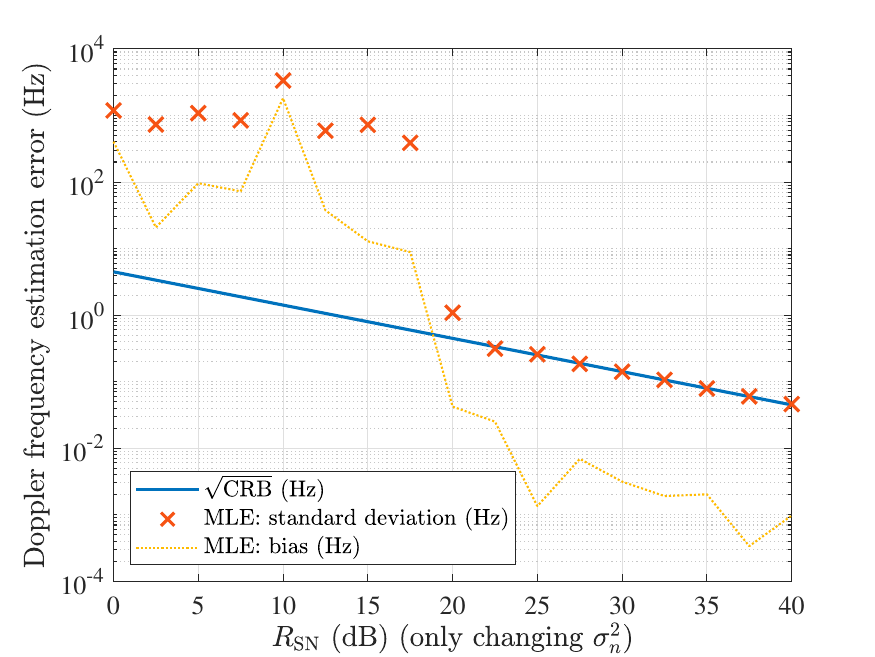}
\caption{{Doppler CRBs (\ref{proposition_aldfikhajkldihf_first_formula}) and the MLE results versus $R_\text{SN}$.}}
\label{Fig_simu_simulation_R_SN}
\end{figure}

Fig. \ref{Fig_simu_simulation_R_SN}
presents the CRB and the estimation results using the maximum likelihood estimation (MLE) algorithm
with respect to $\sigma_n^2$.
Note that
only $\sigma_n^2$ is changed
and
the other parameters remain constant in Fig. \ref{Fig_simu_simulation_R_SN}.
The MLE algorithm is designed based on the approximation in (\ref{R_approximation}).
It can be seen from Fig. \ref{Fig_simu_simulation_R_SN} that
the estimator can be regarded as ``unbiased'',
in the sense that the bias is negligibly small in high SNR regions,
and is {basically overlapping with} the CRB when $R_\text{SN} > 22.5 \text{dB}$.
This verifies the preciseness of the derived CRB in Proposition \ref{proposition_context_CRB_single_target}.

\begin{figure}[!t]
\centering
\includegraphics[width=3.2in]{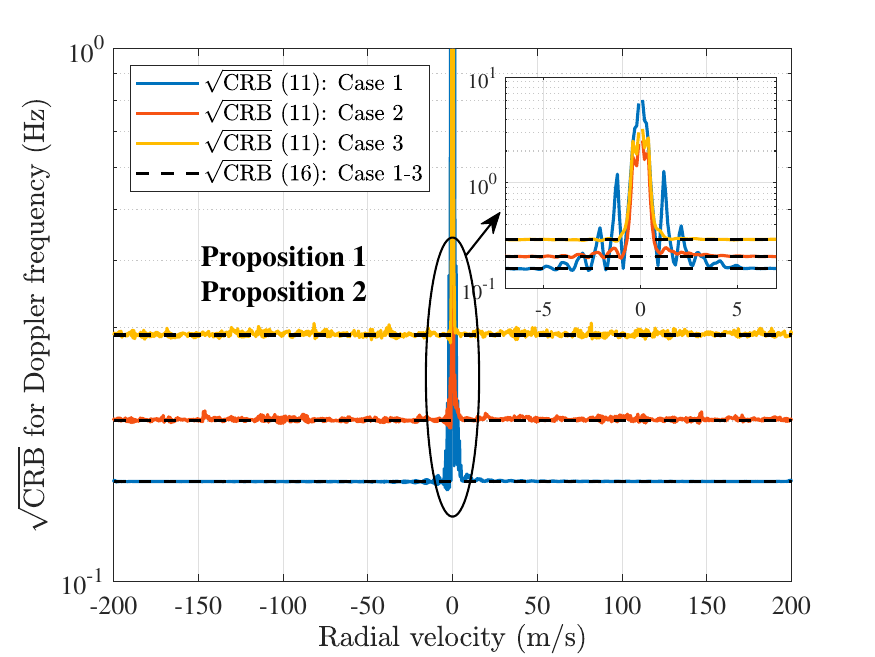}
\caption{Doppler CRBs and the MLE results versus radial velocity,
 comparing CRB results in (\ref{proposition_aldfikhajkldihf_first_formula}) and (\ref{CRB_closed_form_single_text}).}
\label{Fig_simu_CRB_three_approximation}
\end{figure}

Fig. \ref{Fig_simu_CRB_three_approximation}
plots the CRBs of three different cases of $\bm{\varphi} \in \mathbb{Z}^{K \times 1}$
using
(\ref{proposition_aldfikhajkldihf_first_formula})
and
(\ref{CRB_closed_form_single_text}), respectively.
The OFDM symbol index $\bm{\varphi}$ of Case 1 follows (\ref{Optimization_K_even}),
while {$\bm{\varphi}$'s in Cases 2 and 3 are randomly generated.}
From Fig. \ref{Fig_simu_CRB_three_approximation},
we see that Case 1 leads to the minimal CRB in sidelobe regions,
which validate Proposition \ref{proposition_general_optimazition_result}.
We also see that
the
actual
CRB in (\ref{proposition_aldfikhajkldihf_first_formula})
{and the approximated one in}
(\ref{CRB_closed_form_single_text})
are basically overlapping in sidelobe regions with
$\left|f_d \right| > f_{d, \text{M}}$,
{which}
validates the feasibility of using
the CRB approximation
(\ref{CRB_closed_form_single_text})
for the CRB analysis.

\begin{figure*}[!t]
\centering
\subfigure[]
{\includegraphics[width=2.35in]{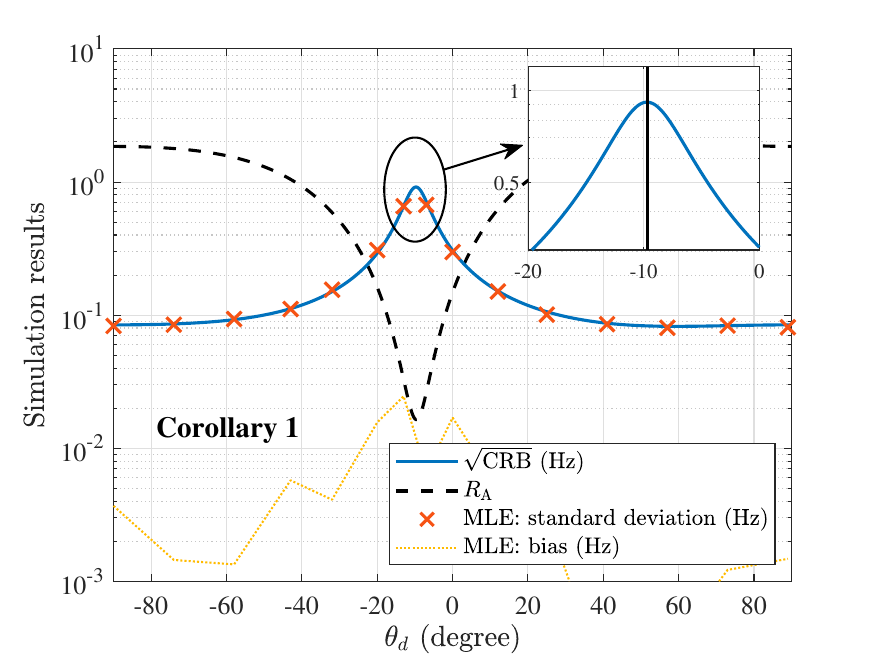}}
\subfigure[]
{\includegraphics[width=2.35in]{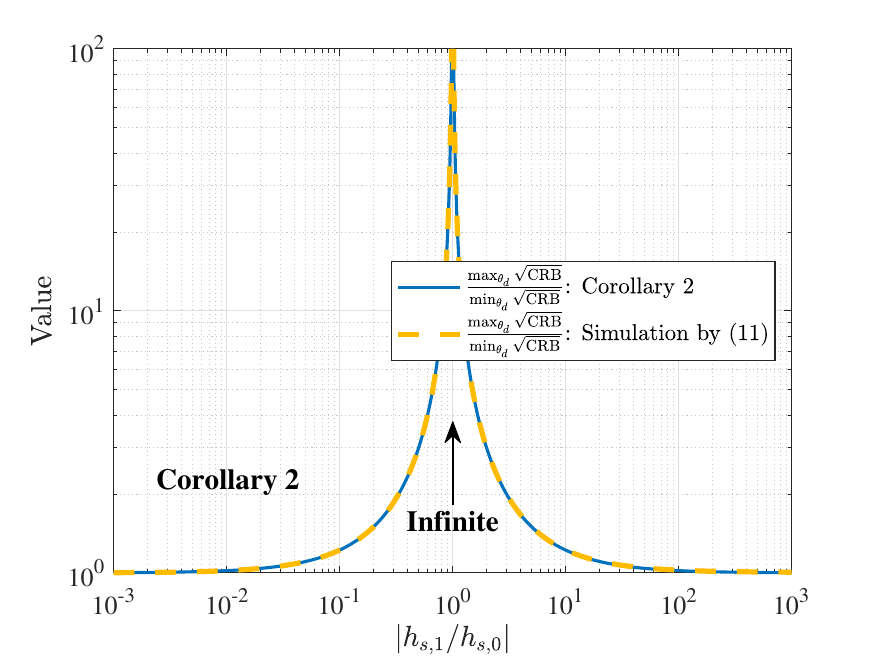}}
\subfigure[]
{\includegraphics[width=2.35in]{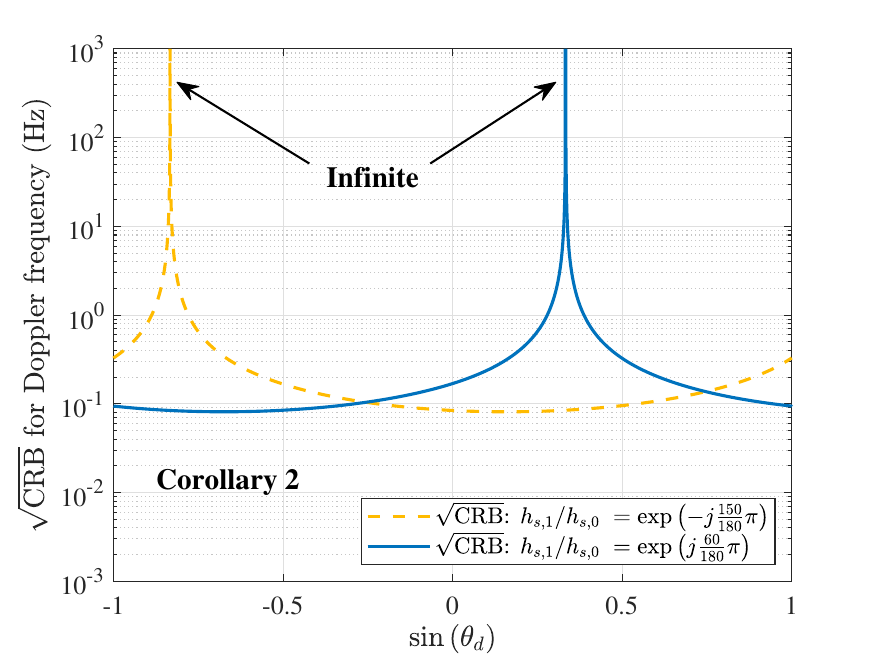}}
\caption{Doppler CRBs (\ref{proposition_aldfikhajkldihf_first_formula}) for Doppler versus $\theta_d$.
(a) Doppler CRBs versus $\theta_d$;
(b) Results of (\ref{corollary_CRB_infinite_FORMULA_c}) versus $\left|{{h_{s, 1}}/{h_{s, 0}}}\right|$;
(c) Doppler CRBs versus $\sin\left(\theta_d\right)$.}
\label{Fig_simu_theta_h1}
\end{figure*}

Fig. \ref{Fig_simu_theta_h1}
demonstrates the influence of $\theta_d$ on Doppler CRBs.
Fig. \ref{Fig_simu_theta_h1} (a)
shows the CRB (given by (\ref{proposition_aldfikhajkldihf_first_formula})), MLE, and $R_\text{A}$ results versus the angle of the dynamic target $\theta_d$.
As illustrated in Corollary \ref{corollary_SSDA_what_what_aaa},
the Doppler CRB is strongly related with $R_\text{A}$.
In Fig. \ref{Fig_simu_theta_h1} (a),
$R_\text{A}$ achieves its lowest value when CRB is maximum,
which verifies the result in Corollary \ref{corollary_SSDA_what_what_aaa}.
In addition, the peak position is closely related to the phase of $h_{s,0}$ and $h_{s,1}$.
The angle corresponding to the CRB peak can be calculated from
the numerator of $R_\text{A}$ (i.e., ${{{\left| {{h}_{s,1}}-a\left( {{\theta }_{d}} \right){{h}_{s,0}} \right|}^{2}}}$),
which is $\theta_d = -9.6^{\circ}$.
This is also validated
in Fig. \ref{Fig_simu_theta_h1} (a).
Moreover, we see the obvious impact
of $\theta_d$ on CRB.
The ratio of the peak and minimum values of the CRB in Fig. \ref{Fig_simu_theta_h1} (a) approximately equals to 11,
i.e.,
$\frac{{{\max }_{{{\theta }_{d}}}}\sqrt{\text{CRB}}}{{{\min }_{{{\theta }_{d}}}}\sqrt{\text{CRB}}} \approx 11$,
{showing the angle-dependent impact of the equivalent static path
on the dynamic path's Doppler estimation in CSI ratio-based sensing.}

Fig. \ref{Fig_simu_theta_h1} (b)
further {observes}
the
ratio of the maximum and minimum values of the CRB in $\theta_d \in \left[-\frac{\pi}{2}, \frac{\pi}{2}\right]$
(i.e., $\frac{{{\max }_{{{\theta }_{d}}}}\sqrt{\text{CRB}}}{{{\min }_{{{\theta }_{d}}}}\sqrt{\text{CRB}}}$
defined in
Corollary
\ref{corollary_CRB_infinite})
to evaluate the variation in CRB values
against $\left|{h_{s, 1}}/{h_{s, 0}}\right|$.
It can be seen that the derivation result in Corollary \ref{corollary_CRB_infinite}
fits well with the simulation result using (\ref{proposition_aldfikhajkldihf_first_formula}).
From Fig. \ref{Fig_simu_theta_h1} (b),
we know that
$\theta_d$ has {negligible} influence on the CRB when
$ \left|{{{h}_{s,1}}}/{{{h}_{s,0}}}\right| \ll 1$ or $ \left|{{{h}_{s,1}}}/{{{h}_{s,0}}}\right| \gg 1$
due to $\frac{{{\max }_{{{\theta }_{d}}}}\sqrt{\text{CRB}}}{{{\min }_{{{\theta }_{d}}}}\sqrt{\text{CRB}}} \approx 1$.
When the value of $\left|{{{h}_{s,0}}}\right|$ is closer to $\left|{{{h}_{s,1}}}\right|$,
the CRB is more sensitive to the changes on $\theta_d$.
An extreme case is $\left|{{{h}_{s,0}}}\right| = \left|{{{h}_{s,1}}}\right|$,
where the ratio value is infinite.

Fig. \ref{Fig_simu_theta_h1} (c)
displays the CRB results when $\left|{{{h}_{s,1}}}/{{{h}_{s,0}}}\right| = 1$ under two different cases.
Clearly,
there is one dynamic target angle that
cannot be sensed using CSI ratios.
Refer to the angle as the unobservable angle.
The angle
is strongly dependent of the phase of ${{{h}_{s,1}}}/{{{h}_{s,0}}}$,
as can be seen in Fig. \ref{Fig_simu_theta_h1} (c) and verified from (\ref{corollary_CRB_infinite_FORMULA_b}) in Corollary \ref{corollary_CRB_infinite}.
Moreover,
each CRB curve in Fig. \ref{Fig_simu_theta_h1} (c) also has a minimum value.
For both curves in Fig. \ref{Fig_simu_theta_h1} (c),
the difference of $\sin \left(\theta_d\right)$ corresponding to the maximum and minimum points of the curve
equals to ${\lambda}/{2d} = 1$, which can be verified from
(\ref{corollary_CRB_infinite_FORMULA_a})
and
(\ref{corollary_CRB_infinite_FORMULA_b}).
{Intuitively, if}
the angle of the dynamic target deviates {further} from the unobservable angle,
the CRB can be lower.

\begin{figure}[!t]
\centering
\includegraphics[width=3.2in]{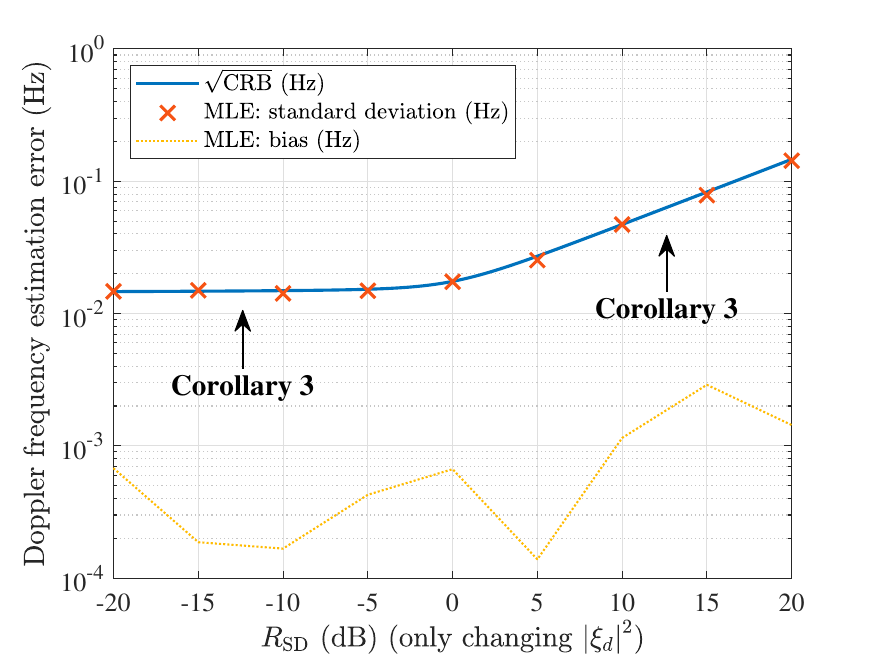}
\caption{Doppler CRBs (\ref{proposition_aldfikhajkldihf_first_formula}) and the MLE results versus $R_\text{SD}$.}
\label{Fig_simu_simulation_R_SD}
\end{figure}

Fig. \ref{Fig_simu_simulation_R_SD}
shows the CRB and the Doppler estimation,
where
we only change the value of $\left|\xi_d\right|^{2}$ to {observe} different $R_\text{SD}$.
When $R_{\text{SD}}$ is sufficiently large,
i.e., the dynamic target power $\left|\xi_d\right|^{2}$ is significantly lower than the averaged static object power
${{\left( {{\left| {{h}_{s,0}} \right|}^{2}}+{{\left| {{h}_{s,1}} \right|}^{2}} \right)}/{2}}$,
a larger $\left|\xi_d\right|^{2}$ brings a more accurate estimation,
as shown in Fig. \ref{Fig_simu_simulation_R_SD}.
This result can be verified using (\ref{ajklnakljfn}) in Corollary \ref{corollary_ksi_large_small}.

\subsection{Evaluation of the CRB Optimization}
Next, we evaluate the proposed noise-limited Doppler sensing method in Proposition \ref{proposition_general_optimazition_result}
and the proposed interference-limited Doppler sensing method in Algorithm \ref{Algorithm1}.

\begin{figure}[!t]
\centering
\includegraphics[width=3.2in]{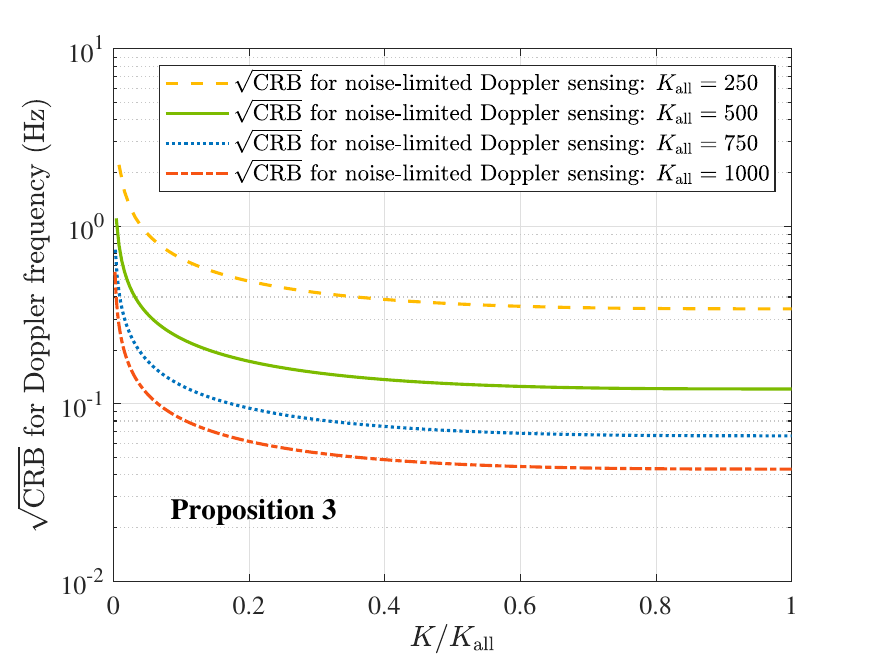}
\caption{Doppler CRB (\ref{proposition_aldfikhajkldihf_first_formula}) of the proposed noise-limited Doppler sensing method in Proposition \ref{proposition_general_optimazition_result} versus ${K}/{K_{\text{all}}}$.}
\label{Fig_simu_K_K_all}
\end{figure}

Fig. \ref{Fig_simu_K_K_all}
shows the CRB of the proposed noise-limited Doppler sensing method in Proposition \ref{proposition_general_optimazition_result} under different ${K}/{K_{\text{all}}}$.
These results are the minimum values that the CRB can achieve
with given ${K}/{K_\text{all}}$ and $K_\text{all}$ when
{$\left|f_d \right| > f_{d, \text{M}}$}.
As can be seen, the CRBs decrease rapidly when ${K}/{K_{\text{all}}}$ is small.
However,
when ${K}/{K_{\text{all}}}$ is larger than 0.5, increasing ${K}/{K_{\text{all}}}$
has limit improvement to the accuracy.
In this case,
it is more feasible to expand the available range of the OFDM symbol $K_{\text{all}}$
to further enhance the CRB.
Although the noise-limited sensing result has the best performance when
$\left|f_d \right| > f_{d, \text{M}}$,
it can be observed from
the Case 1 in Fig. \ref{Fig_simu_CRB_three_approximation} (b) that
the CRB has significant fluctuation when
the velocity is close to zero, i.e., $\left| v \right|<5{\text{m}}/{\text{s}}$.
Therefore,
a further optimization is required
to reduce the mainlobe of the Doppler pattern,
hence, the motivation of Algorithm \ref{Algorithm1}.

\begin{figure}[!t]
\centering
\includegraphics[width=3.2in]{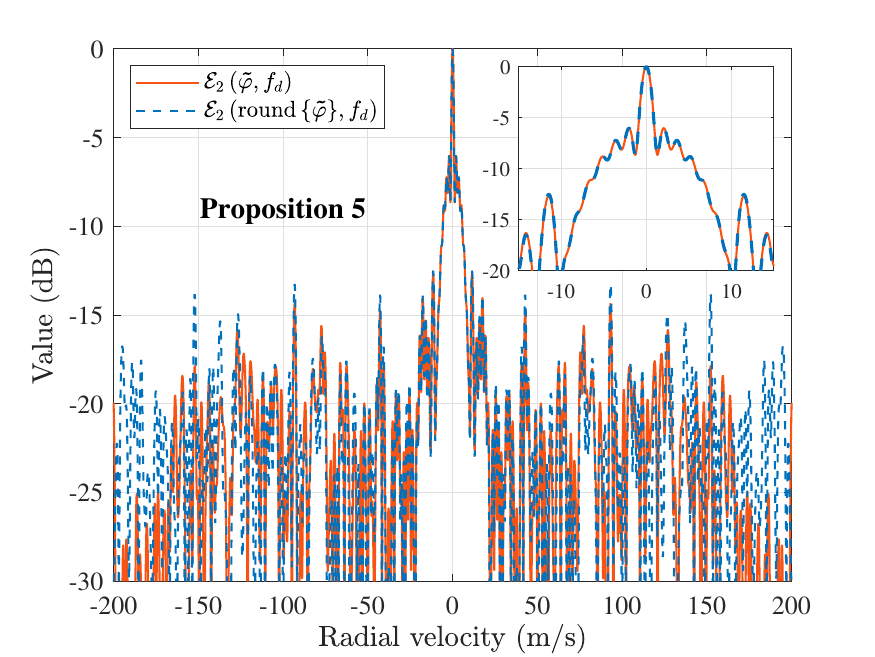}
\caption{Comparison of ${\mathcal{E}_\mathcal{P}}\left( {{\bm{\tilde{\varphi} }}}, f_d \right)$
and
${\mathcal{E}_\mathcal{P}}\left( {\text{round}\left\{{\bm{\tilde{\varphi} }}\right\}}, f_d \right)$.}
\label{Fig_simu_round_result}
\end{figure}

Fig. \ref{Fig_simu_round_result}
evaluates the impact of
the small-amount perturbation of ${{\bm{\tilde{\varphi} }}}$
on the envelope ${\mathcal{E}_\mathcal{P}}\left( {{\bm{\tilde{\varphi} }}}, f_d \right)$,
as illustrated in Proposition \ref{proposition_integer_operation}.
We see that the mainlobe and nearby sidelobes of
${\mathcal{E}_\mathcal{P}}\left( {{\bm{\tilde{\varphi} }}}, f_d \right)$,
have minimal changes with or without the perturbation on ${{\bm{\tilde{\varphi} }}}$.
This
verifies Proposition \ref{proposition_integer_operation}
and the feasibility of Algorithm \ref{Algorithm1}.

\begin{figure}[!t]
\centering
\subfigure[]
{\includegraphics[width=3.2in]{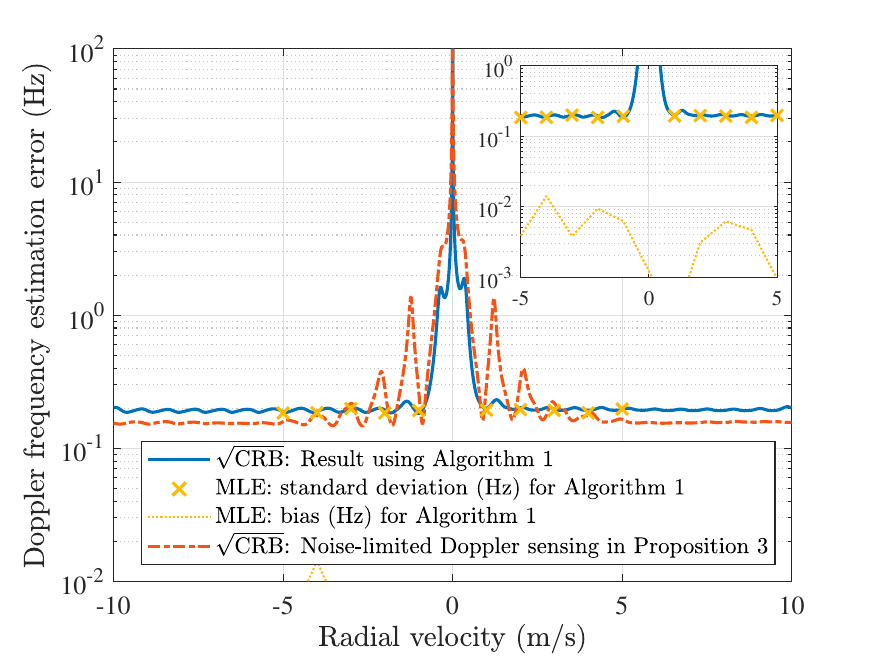}}
\subfigure[]
{\includegraphics[width=3.2in]{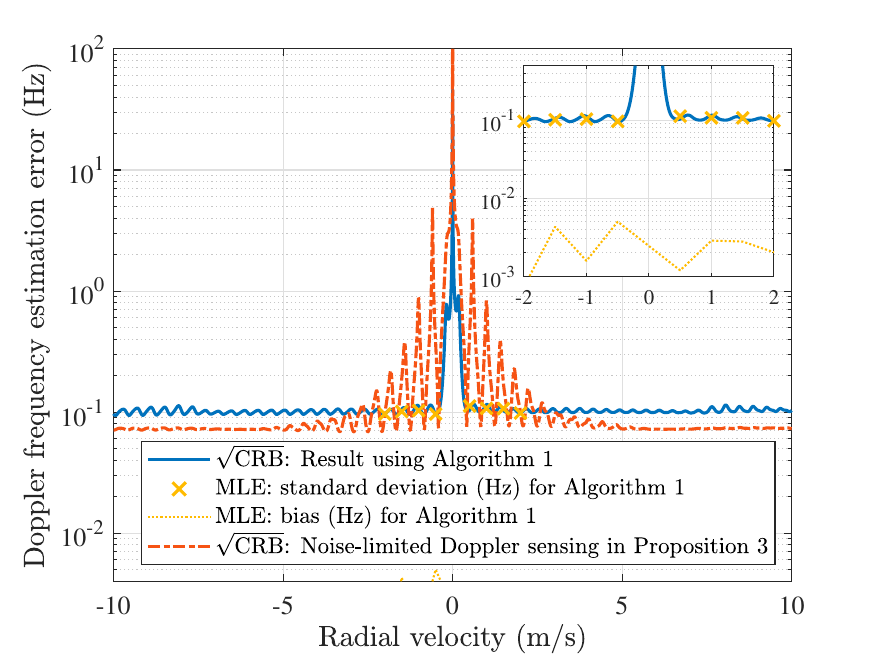}}
\caption{Comparison of the noise-limited CRB result using Proposition \ref{proposition_general_optimazition_result},
the interference-limited CRB result using the proposed Algorithm \ref{Algorithm1},
and the MLE results for Algorithm \ref{Algorithm1},
with
(a) $K_\text{all} = 512$, $K = 128$, and $K_F = 16$,
and
(b) $K_\text{all} = 1024$, $K = 128$, and $K_F = 16$.}
\label{Fig_simu_512_1024_all_result}
\end{figure}

Fig. \ref{Fig_simu_512_1024_all_result}
displays the interference-limited Doppler sensing CRB using Algorithm \ref{Algorithm1},
where $K_\text{all} = 512$ and $1024$ for Figs. \ref{Fig_simu_512_1024_all_result} (a) and \ref{Fig_simu_512_1024_all_result} (b), respectively.
As benchmark performances, Figs. \ref{Fig_simu_512_1024_all_result} (a)
and
\ref{Fig_simu_512_1024_all_result} (b)
plot the CRB under the noise-limited Doppler sensing method given in Proposition \ref{proposition_general_optimazition_result}.
We see that
both the benchmark
CRB results have large values
when
$\left| v \right|<5{\text{m}}/{\text{s}}$,
{indicating poor sensing performances of a}
dynamic target.
Algorithm \ref{Algorithm1} is proposed to reduce the
mainlobe of the envelope ${\mathcal{E}_\mathcal{P}}\left( {{\bm{\tilde{\varphi} }}}, f_d \right)$
to reduce the CRB.
Fig. \ref{Fig_simu_512_1024_all_result}
validates the effectiveness of Algorithm \ref{Algorithm1}.
We see that the CRB values are substantially reduces between
$1{\text{m}}/{\text{s}}<\left| v \right|<5{\text{m}}/{\text{s}}$
in Fig. \ref{Fig_simu_512_1024_all_result} (a) and
between
$0.5{\text{m}}/{\text{s}}<\left| v \right|<5{\text{m}}/{\text{s}}$
in Fig. \ref{Fig_simu_512_1024_all_result} (b).
The CRB values are increased in the sidelobe regions but only slightly.
Moreover,
the observable velocity range for the cases $K = K_{\text{all}} = 512$ and $K = K_{\text{all}} = 1024$ is $\left| v \right|>0.8{\text{m}}/{\text{s}}$ and $\left| v \right|>0.4{\text{m}}/{\text{s}}$, respectively,
as can be seen in Fig. \ref{Fig_simu_512_1024_all_result}.
We also provide the MLE results in Fig. \ref{Fig_simu_512_1024_all_result}.
Results show that
the interference-limited Doppler sensing CRB using Algorithm \ref{Algorithm1}
is not only the theoretical CRB but also achievable in both figures.

\begin{figure}[!t]
\centering
\subfigure[]
{\includegraphics[width=3.2in]{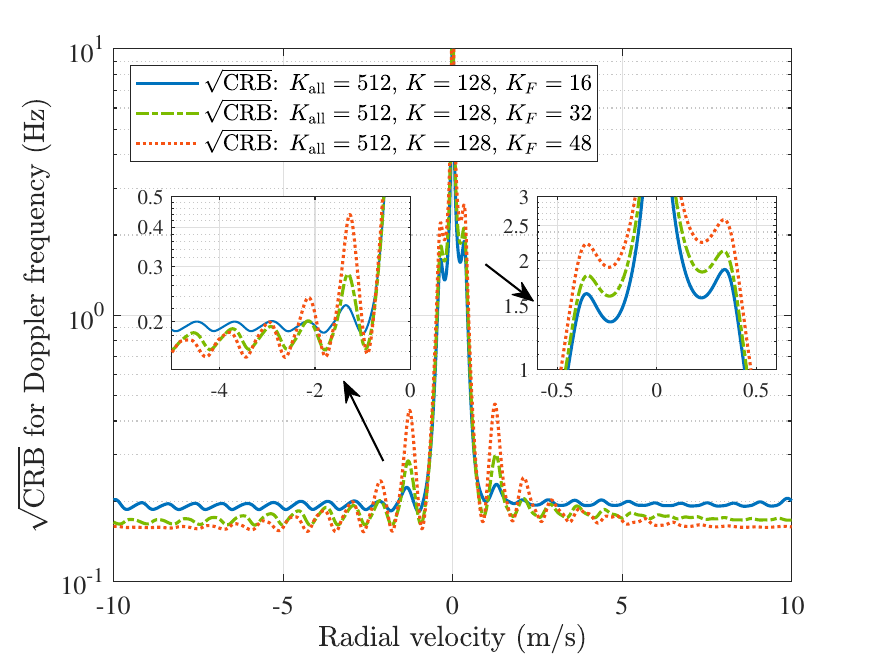}}
\subfigure[]
{\includegraphics[width=3.2in]{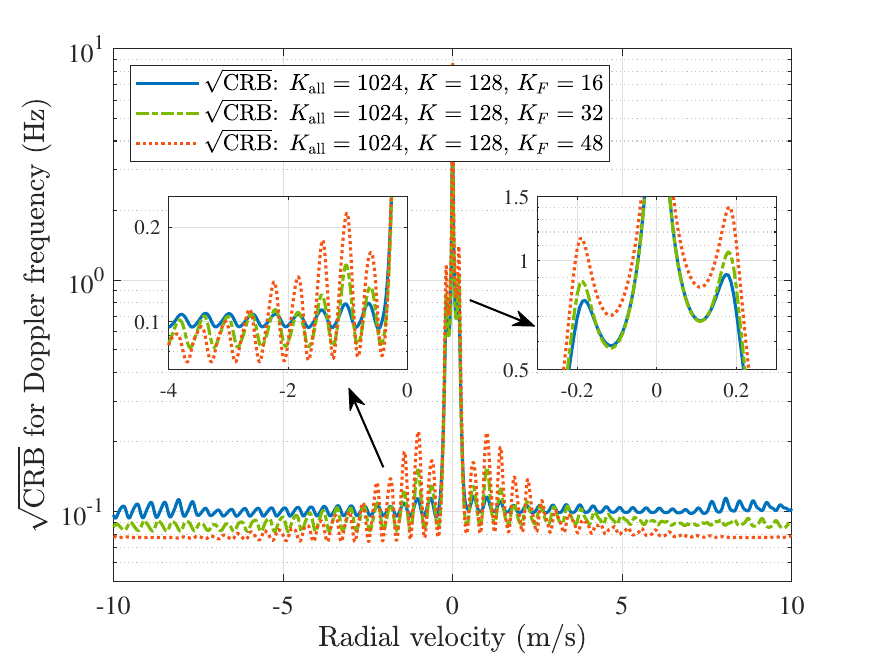}}
\caption{Comparison of the optimal CRB results using proposed Algorithm \ref{Algorithm1} under different $K_F$:
(a) $K_\text{all} = 512$, $K = 128$;
(b) $K_\text{all} = 1024$, $K = 128$.}
\label{Fig_simu_512_1024_multiple_result}
\end{figure}

Fig. \ref{Fig_simu_512_1024_multiple_result}
displays the CRB results under different $K_F$.
Recall that the term $K_F$ denotes the number of the fixed OFDM symbols which are not used in the
waveform optimization; see also Fig. \ref{Fig_text_optimization_figure}.
From Fig. \ref{Fig_simu_512_1024_multiple_result},
we see that,
when $K_F$ is reduced,
the sidelobe of the CRB increases,
while the mainlobe level increases.
This is expected,
since the fixed OFDM symbols are used to ensure that
the CRB
in the sidelobe region of the Doppler pattern is low.
On the contrary,
less fixed OFDM symbols lead to
more symbols available for optimization,
as illustrated in Fig. \ref{Fig_text_optimization_figure}.
Therefore, the mainlobe level of the CRB decreases.
Fig. \ref{Fig_simu_512_1024_multiple_result} validates that
$K_F$ provides an extra degree of freedom for handling practical requirements in waveform optimization.

\section{Conclusion}\label{Section_Conclusion}
In this paper,
we derive the CRB of the CSI-ratio model in high-SNR and single-dynamic-target case.
For simplifying the analysis and the optimization,
we have provided a concise closed-form approximation of the Doppler CRB,
using which we analyze the influence of each parameter on the model accuracy.
Then, we enhance the accuracy of the Doppler CRB by optimizing the OFDM symbol index.
We first propose the noise-limited Doppler sensing method with the Doppler CRB optimization.
However, the model accuracy for the dynamic target with low-Doppler is obviously poor.
We have then proposed an OFDM symbol optimization algorithm in Algorithm \ref{Algorithm1} to reduce the CRB at low-Doppler range.
Simulation results verify the correctness of the derivation and the analysis.

For future work,
we can improve the CSI ratio performance by selecting the best pairs of antennas and the best subcarriers for sensing
since the CRB result is related to the antenna and subcarrier indexes.
We can also extend the work to the case of joint sensing using multiple subcarriers.

\begin{appendices}

\section{Proof of Proposition \ref{proposition_aldfikhajkldihf}}\label{Appendix_aldfikhajkldihf}
The partial derivative of $\bm{\chi }$ with respect to each element in $\bm{\alpha}$,
as defined in (\ref{alpha_parameters_unknown}), can be calculated as:
\begin{align}\label{single_target_diff}
\nonumber  \frac{\partial \bm{\chi }}{\partial {{f}_{d}}}& =j {2\pi {{T}_{0}}}{{\rho }_{1}}\text{diag}\left\{ {\bm{\varphi }} \right\}\frac{\mathbf{d}\left( {{f}_{d}} \right)\left[ a\left( {{\theta }_{d}} \right)-{{\rho }_{0}} \right]}{{{\left( {{\rho }_{1}}\mathbf{d}\left( {{f}_{d}} \right)+{{\mathbf{1}}_{K\times 1}} \right)}^{2}}} \\
\nonumber \frac{\partial \bm{\chi }}{\partial {{\theta }_{d}}}&={{\rho }_{1}}\frac{\partial a\left( {{\theta }_{d}} \right)}{\partial {{\theta }_{d}}}\frac{\mathbf{d}\left( {{f}_{d}} \right)}{{{\rho }_{1}}\mathbf{d}\left( {{f}_{d}} \right)+{{\mathbf{1}}_{K\times 1}}} \\
\nonumber \frac{\partial \bm{\chi }}{\partial \operatorname{Re}\left\{ {{\rho }_{0}} \right\}}&=\frac{{{\mathbf{1}}_{K\times 1}}}{{{\rho }_{1}}\mathbf{d}\left( {{f}_{d}} \right)+{{\mathbf{1}}_{K\times 1}}}=-j\frac{\partial \bm{\chi }}{\partial \operatorname{Im}\left\{ {{\rho }_{0}} \right\}} \\
 \frac{\partial \bm{\chi }}{\partial \operatorname{Re}\left\{ {{\rho }_{1}} \right\}}&=\frac{\mathbf{d}\left( {{f}_{d}} \right)\left[ a\left( {{\theta }_{d}} \right)-{{\rho }_{0}} \right]}{{{\left( {{\rho }_{1}}\mathbf{d}\left( {{f}_{d}} \right)+{{\mathbf{1}}_{K\times 1}} \right)}^{2}}}=-j\frac{\partial \bm{\chi }}{\partial \operatorname{Im}\left\{ {{\rho }_{1}} \right\}}.
\end{align}

Substituting (\ref{single_target_diff}) into (\ref{CRB_for}), we can obtain the FIM of the parameters, denoted as $\mathbf{F}$.
Utilizing the Schur complement, the inverse of the FIM can
be expressed as:
\begin{align}\label{laikjfalijoi}
\nonumber &\mathbf{C}\left( \bm{\alpha } \right)= \mathbf{F}^{-1} \\
 & \triangleq {{\left[ \begin{matrix}
   {{\mathbf{F}}_{11}}\in {{\mathbb{R}}^{2\times 2}} & {{\mathbf{F}}_{12}}\in {{\mathbb{R}}^{2\times 4}}  \\
   \mathbf{F}_{12}^{T}\in {{\mathbb{R}}^{4\times 2}} & {{\mathbf{F}}_{22}}\in {{\mathbb{R}}^{4\times 4}}  \\
\end{matrix} \right]}^{-1}}=\left[ \begin{array}{c;{2pt/2pt}c}
   {{\mathbf{H}}^{-1}} & \times   \\
   \hdashline[2pt/2pt]
   \times  & \times   \\
\end{array} \right],
\end{align}
where
$\mathbf{H}={{\mathbf{F}}_{11}}-{{\mathbf{F}}_{12}}\mathbf{F}_{22}^{-1}\mathbf{F}_{12}^{T}$.
Using (\ref{laikjfalijoi}),
the result given in (\ref{cal_H_result}) can be obtained.
Proposition \ref{proposition_aldfikhajkldihf} is thus proved.

\section{Proof of Proposition \ref{proposition_context_CRB_single_target}}\label{Appendix_proposition_context_CRB_single_target}
{From the CRB derived in Appendix \ref{Appendix_aldfikhajkldihf},
there are nine key terms, i.e.,}
${{\varepsilon }_{1}}$,
${{\varepsilon }_{K,1}}$,
${{\varepsilon }_{K,2}}$,
${{\varepsilon }_{\mu ,1}}$,
${{\varepsilon }_{\mu ,2}}$,
${{\varepsilon }_{d,\mu ,2}}$,
${{\varepsilon }_{d,\mu ,1}}$,
${{\varepsilon }_{K,\mu }}$,
and
${{\varepsilon }_{d,K,\mu }}$,
as given in (\ref{proposition_variables_definition}). Here,
we derive the approximation of ${{\varepsilon }_{1}}$ as an example.
Before further derivations, we provide the following lemma,
with its proof established in Appendix \ref{Appendix_Lemma_2}.

\vspace{0.3em}

\begin{lemma}\label{Lemma_X_b0}
Let $\phi \in \left[-\pi, \pi\right]$ be a constant.
When {$\left|f_d \right| > f_{d, \text{M}}$ and} $b_0 > 2b_1$, we have the following approximation:
\begin{align}\label{Lemma_X_b0_b1_equation_all}
\mathcal{L}_{{{f}_{d}}}\left( b_0, b_1 \right)&=\sum\limits_{k=0}^{K-1}{\frac{1}{{{b}_{0}}+{{b}_{1}}\cos \left( 2\pi {{\varphi }_{k}}{{T}_{0}}{{f}_{d}}+\phi  \right)}} \approx {\frac{K}{\sqrt{b_{0}^{2}-b_{1}^{2}}}}.
\end{align}
\end{lemma}

\vspace{0.3em}

Based on the definition of $\varepsilon_1$ in (\ref{proposition_variables_definition}),
we have:
\begin{align}
\nonumber  & {{\varepsilon }_{1}}=\sum\limits_{k=0}^{K-1}{\frac{1}{{{\left| {{\rho }_{1}}{{e}^{j2\pi k{{f}_{d}}{{T}_{0}}}}+1 \right|}^{2}}+{{\left| {{\rho }_{1}}{{e}^{j2\pi k{{f}_{d}}{{T}_{0}}}}+{{\rho }_{0}}{{a}^{*}}\left( {{\theta }_{d}} \right) \right|}^{2}}}} \\
\nonumber   & =\sum\limits_{k=0}^{K-1}{\frac{1}{{{\left| {{\rho }_{0}} \right|}^{2}}+2{{\left| {{\rho }_{1}} \right|}^{2}}+1+2\operatorname{Re}\left\{ \left[ 1+\rho _{0}^{*}a\left( {{\theta }_{d}} \right) \right]{{\rho }_{1}}{{e}^{j2\pi k{{f}_{d}}{{T}_{0}}}} \right\}}}.
\end{align}
Using Lemma \ref{Lemma_X_b0}, we can obtain
$b_{0}^{\left( {{\varepsilon }_{1}} \right)}={{\left| {{\rho }_{0}} \right|}^{2}}+2{{\left| {{\rho }_{1}} \right|}^{2}}+1$;
$b_{1}^{\left( {{\varepsilon }_{1}} \right)}=2\left| {{\rho }_{1}}+\rho _{0}^{*}{{\rho }_{1}}a\left( {{\theta }_{d}} \right) \right|$.
Moreover, we can show that
$b_{0}^{\left( {{\varepsilon }_{1}} \right)} > b_{1}^{\left( {{\varepsilon }_{1}} \right)}$
always holds, specifically:
\begin{align}
\nonumber b_{0}^{\left( {{\varepsilon }_{1}} \right)} - b_{1}^{\left( {{\varepsilon }_{1}} \right)}& ={{\left| {{\rho }_{0}} \right|}^{2}}+2{{\left| {{\rho }_{1}} \right|}^{2}}+1-2\left| \left[ 1+\rho _{0}^{*}a\left( {{\theta }_{d}} \right) \right]{{\rho }_{1}} \right| \\
\nonumber & \ge {{\left| {{\rho }_{0}} \right|}^{2}}+2{{\left| {{\rho }_{1}} \right|}^{2}}+1-2\left( \left| {{\rho }_{0}} \right|+1 \right)\left| {{\rho }_{1}} \right| \\
 & ={{\left( \left| {{\rho }_{1}} \right|-1 \right)}^{2}}+{{\left( \left| {{\rho }_{0}} \right|-\left| {{\rho }_{1}} \right| \right)}^{2}}>0.
\end{align}
To meet the accuracy requirement of the approximation, we need to ensure that
$b_{0}^{\left( {{\varepsilon }_{1}} \right)} > 2b_{1}^{\left( {{\varepsilon }_{1}} \right)}$ based on Lemma \ref{Lemma_X_b0},
i.e.:
\begin{align}\label{sdolweoidkdfn}
{{\left| {{\rho }_{1}} \right|}^{2}}\le \frac{1}{16}\left( 1+{{\left| {{\rho }_{0}} \right|}^{2}} \right) \ \ \text{or} \ \ {{\left| {{\rho }_{1}} \right|}^{2}}\ge {5}\left( 1+{{\left| {{\rho }_{0}} \right|}^{2}} \right).
\end{align}
Therefore,
when $\left|f_d \right| > f_{d, \text{M}}$
and
(\ref{sdolweoidkdfn}) are satisfied,
the approximation of ${{\varepsilon }_{1}}$ is:
\begin{align}
 {{\varepsilon }_{1}}\approx {{\mathcal{L}}_{{{f}_{d}}}}\left( b_{0}^{\left( {{\varepsilon }_{1}} \right)},b_{1}^{\left( {{\varepsilon }_{1}} \right)} \right),
\end{align}

The other eight terms can be approximated using the similar method devised for ${{\varepsilon }_{1}}$.
With the details suppressed, we provide the following results for $\left|f_d \right| > f_{d, \text{M}}$:
\begin{align}\label{relationship_Appendix_A}
\nonumber  {{\varepsilon }_{K,1}}&\approx {{S}_{\varphi ,1}}\left( 2\pi {{T}_{0}} \right){{\varepsilon }_{1}}, \ &{{\varepsilon }_{K,2}}&\approx {{S}_{\varphi ,2}}{{\left( 2\pi {{T}_{0}} \right)}^{2}}{{\varepsilon }_{1}}, \\
\nonumber  {{\varepsilon }_{K,\mu }}&\approx {{S}_{\varphi ,1}}\left( 2\pi {{T}_{0}} \right){{\varepsilon }_{\mu ,1}}, \ &{{\varepsilon }_{d,K,\mu }}&\approx {{S}_{\varphi ,1}}\left( 2\pi {{T}_{0}} \right){{\varepsilon }_{d,\mu ,1}}, \\
 {{\varepsilon }_{d,\mu ,2}}&\approx {\varepsilon _{\mu ,1}^{*}{{\varepsilon }_{d,\mu ,1}}}/{{{\varepsilon }_{1}}},
\end{align}
where
${{S}_{\varphi ,1}}=\frac{1}{K}\mathbf{1}_{K\times 1}^{T}{\bm{\varphi }}$,
and
${{S}_{\varphi ,2}}=\frac{1}{K}{\bm{\varphi }}^{T}{\bm{\varphi }}$.

Utilizing (\ref{relationship_Appendix_A}),
the terms ${{\left[ \mathbf{H} \right]}_{1,1}}$ and ${{\left[ \mathbf{H} \right]}_{1,2}}$ can be approximated to
${{\left[ \mathbf{H} \right]}_{1,1}}\approx 8{{\pi }^{2}}T_{0}^{2}\frac{{{\left| {{\xi }_{d}}{{\rho }_{2}} \right|}^{2}}}{\sigma _{n}^{2}} \left( {{S}_{\varphi ,2}}-S_{\varphi ,1}^{2} \right){{\varepsilon }_{1}}$
and
${{\left[ \mathbf{H} \right]}_{1,2}}\approx 0$, respectively.
Hence, the inverse of the closed-form CRB for Doppler frequency can {be calculated as}:
\begin{align}\label{Appendix_aaa_final_result}
\nonumber &\text{CRB}_{{{f}_{d}}}^{-1}\approx 8{{\pi }^{2}}T_{0}^{2}\frac{{{\left| a\left( {{\theta }_{d}} \right)-{{\rho }_{0}} \right|}^{2}}}{\sigma _{n}^{2}{{\left| {{\xi }_{d}} \right|}^{-2}}}\left( {{S}_{\varphi ,2}}-S_{\varphi ,1}^{2} \right){{\varepsilon }_{1}} \\
&= 8{{\pi }^{2}}T_{0}^{2}K\left( {{S}_{\varphi ,2}}-S_{\varphi ,1}^{2} \right)  \frac{{{R}_{\text{SN}}}{{R}_{\text{A}}}}{\sqrt{{{\left( 1-{{R}_{\text{SD}}} \right)}^{2}}+2{{R}_{\text{A}}}{{R}_{\text{SD}}}}},
\end{align}
where
$R_\text{SN}$,
$R_\text{SD}$,
and
$R_\text{A}$
are defined in Proposition \ref{proposition_context_CRB_single_target}.

\section{proof of Lemma \ref{Lemma_X_b0}}\label{Appendix_Lemma_2}

Utilizing the property of the convolution operation,
(\ref{Lemma_X_b0_b1_equation_all}) can be expressed as:
\begin{align}\label{proof_lemma_X_b0_b1_ccc}
\mathcal{L}_{{{f}_{d}}}\left( b_0, b_1 \right) &=\frac{1}{{{K}_{\text{all}}}}\sum\limits_{\tilde{k}=0}^{{{K}_{\text{all}}}-1}{{{X}_{c}}\left( {\tilde{k}} \right)\sum\limits_{k=0}^{K-1}{{{e}^{j\frac{2\pi }{{{K}_{\text{all}}}}{{\varphi }_{k}}\tilde{k}}}}},
\end{align}
where
${{X}_{c}}\left( {\tilde{k}} \right)={{\mathcal{F}}_{{\tilde{k}}}}\left\{ \frac{1}{{{b}_{0}}+{{b}_{1}}\cos \left( 2\pi k{{T}_{0}}{{f}_{d}}+\phi  \right)} \right\}$,
${\tilde{k}} = 0,\cdots,K_{\text{all}} - 1$.
The approximation of ${{X}_{c}}\left( {\tilde{k}} \right)$ is given in (\ref{Lemma_X_b0_b1_equation_bbb_proof}) at the top of the next page under the assumption $b_0 > 2b_1$,
where
$\delta \left( {\tilde{k}} \right) = 1$ if ${\tilde{k}} = 0$
and
$\delta \left( {\tilde{k}} \right) = 0$ if ${\tilde{k}} \ne 0$.
In the first step of (\ref{Lemma_X_b0_b1_equation_bbb_proof}), we use the relationship
${{\left( 1+x \right)}^{-1}}=\sum\nolimits_{n=0}^{\infty }{{{\left( -1 \right)}^{n}}{{x}^{n}}}$.
In the second step,
we expand ${{\cos }^{n}}\left( \cdot  \right)$ using the binomial theorem.
We then split the second step into two terms according to the parity of $n$ and get the third step.
Then, a key approximation is made in the third step by remaining $r = n$ in the first term
and $r = n \pm 1$ in the second term based on the assumption $b_0 > 2b_1$, and we get the fourth step.
In the fifth step, the relationship
${{\left( 1-x \right)}^{-\frac{1}{2}}}=\sum\nolimits_{n=0}^{\infty }{{{\left( \frac{1}{4}x \right)}^{n}}C_{2n}^{n}}$ is utilized.

\begin{figure*}[!t]
\begin{align}\label{Lemma_X_b0_b1_equation_bbb_proof}
\nonumber  {{X}_{c}}\left( {\tilde{k}} \right)=&\sum\limits_{n=0}^{\infty }{\frac{{{\left( -{{b}_{1}} \right)}^{n}}}{b_{0}^{n+1}}{{\mathcal{F}}_{{\tilde{k}}}}\left\{ {{\cos }^{n}}\left( 2\pi k{{T}_{0}}{{f}_{d}}+\phi  \right) \right\}}=\sum\limits_{n=0}^{\infty }{\frac{{{\left( -{{b}_{1}} \right)}^{n}}}{{{2}^{n}}b_{0}^{n+1}}\sum\limits_{r=0}^{n}{{C_{n}^{r}}{{\mathcal{F}}_{{\tilde{k}}}}\left\{ {{e}^{j\left( 2r-n \right)\left( 2\pi k{{T}_{0}}{{f}_{d}}+\phi  \right)}} \right\}}} \\
\nonumber =&\sum\limits_{n=0}^{\infty }{\frac{{{\left( -{{b}_{1}} \right)}^{2n}}}{{{2}^{2n}}b_{0}^{2n+1}}\sum\limits_{r=0}^{2n}{{C_{2n}^{r}}{{\mathcal{F}}_{{\tilde{k}}}}\left\{ {{e}^{j\left( 2r-2n \right)\left( 2\pi k{{T}_{0}}{{f}_{d}}+\phi  \right)}} \right\}}}+\sum\limits_{n=0}^{\infty }{\frac{{{\left( -{{b}_{1}} \right)}^{2n+1}}}{{{2}^{2n+1}}b_{0}^{2n+2}}\sum\limits_{r=0}^{2n+1}{{C_{2n+1}^{r}}{{\mathcal{F}}_{{\tilde{k}}}}\left\{ {{e}^{j\left( 2r-2n-1 \right)\left( 2\pi k{{T}_{0}}{{f}_{d}}+\phi  \right)}} \right\}}}\\
\nonumber \approx & \sum\limits_{n=0}^{\infty }{\frac{b_{1}^{2n}}{{{2}^{2n}}b_{0}^{2n+1}}\sum\limits_{r=n}^{n}{{C_{2n}^{r}}{{\mathcal{F}}_{{\tilde{k}}}}\left\{ {{e}^{j\left( 2r-2n \right)\left( 2\pi k{{T}_{0}}{{f}_{d}}+\phi  \right)}} \right\}}}-\sum\limits_{n=0}^{\infty }{\frac{b_{1}^{2n+1}}{{{2}^{2n+1}}b_{0}^{2n+2}}\sum\limits_{r=n}^{n+1}{{C_{2n+1}^{r}}{{\mathcal{F}}_{{\tilde{k}}}}\left\{ {{e}^{j\left( 2r-2n-1 \right)\left( 2\pi k{{T}_{0}}{{f}_{d}}+\phi  \right)}} \right\}}} \\
 = &\: \frac{{{K}_{\text{all}}}}{\sqrt{b_{0}^{2}-b_{1}^{2}}}\delta \left( {\tilde{k}} \right)-{{\mathcal{F}}_{{\tilde{k}}}}\left\{ \cos \left( 2\pi k{{T}_{0}}{{f}_{d}}+\phi  \right) \right\}\sum\limits_{n=0}^{\infty }{\frac{b_{1}^{2n+1}}{{{2}^{2n}}b_{0}^{2n+2}}{C_{2n+1}^{n}}}.
\end{align}
\normalsize
\hrulefill
\end{figure*}

Substituting (\ref{Lemma_X_b0_b1_equation_bbb_proof}) into (\ref{proof_lemma_X_b0_b1_ccc}),
we have:
\begin{align}\label{Appendix_L_ful_formula}
 \nonumber &\mathcal{L}_{{{f}_{d}}}\left( b_0, b_1 \right) \approx \frac{K}{\sqrt{b_{0}^{2}-b_{1}^{2}}} \\
&-K \operatorname{Re}\left\{\frac{{e}^{j\phi }}{K}\sum\limits_{k=0}^{K-1}{{{e}^{j2\pi {{\varphi }_{k}}{{T}_{0}}{{f}_{d}}}}} \right\}\sum\limits_{n=0}^{\infty }{\frac{b_{1}^{2n+1}}{{{2}^{2n}}b_{0}^{2n+2}}{C_{2n+1}^{n}}},
\end{align}
where $\operatorname{Re}\left\{ \frac{{e}^{j\phi }}{K}\sum\nolimits_{k=0}^{K-1}{{{e}^{j2\pi {{\varphi }_{k}}{{T}_{0}}{{f}_{d}}}}} \right\}$
is related with $\mathcal{P}\left( {\bm{\varphi }}, f_d \right)$ defined in Section \ref{Section_approximate_CRB}.
It can be easily proved that
$\operatorname{Re}\left\{ \frac{{e}^{j\phi }}{K}\sum\nolimits_{k=0}^{K-1}{{{e}^{j2\pi {{\varphi }_{k}}{{T}_{0}}{{f}_{d}}}}} \right\} \le \mathcal{P}\left( {\bm{\varphi }}, f_d \right)$.
Therefore, the mainlobe width of $\mathcal{P}\left( {\bm{\varphi }}, f_d \right)$
can be utilized to describe the mainlobe of $\operatorname{Re}\left\{ \frac{{e}^{j\phi }}{K}\sum\nolimits_{k=0}^{K-1}{{{e}^{j2\pi {{\varphi }_{k}}{{T}_{0}}{{f}_{d}}}}} \right\}$.
Focusing on the sidelobe, which is $\left|f_d \right| > f_{d, \text{M}}$,
the second term of (\ref{Appendix_L_ful_formula}) can be eliminated, and (\ref{Lemma_X_b0_b1_equation_all}) is obtained.

\section{Proof of Proposition \ref{Lemma_Envelope_in_aoifhnaol}}\label{Appendix_of_Lemma_Envelope_in_aoifhnaol}
Substituting $\bm{\varphi } \! = \! {{\left[ {{{\bm{\tilde{\varphi }}}}^{T}},{{K}_{\text{all}}} \! - \! 1 - \! {{{\bm{\tilde{\varphi }}}}^{T}} \right]}^{T}}$
into
{$\mathcal{P}\left( {\bm{\varphi }}, f_d \right)$}
yields:
\begin{align}\label{alidkfjaldikhjflahf}
\nonumber & {\mathcal{P}\left( {\bm{\varphi }}, f_d \right)}\\
\nonumber  & =\frac{1}{K}\left| \sum\limits_{k=0}^{\frac{K}{2}-1}{{{e}^{j2\pi {{\left[ {\bm{\tilde{\varphi }}} \right]}_{k}}{{T}_{0}}{{f}_{d}}}}}+{{e}^{j2\pi \left( {{K}_{\text{all}}}-1 \right){{T}_{0}}{{f}_{d}}}}\sum\limits_{k=0}^{\frac{K}{2}-1}{{{e}^{-j2\pi {{\left[ {\bm{\tilde{\varphi }}} \right]}_{k}}{{T}_{0}}{{f}_{d}}}}} \right| \\
\nonumber  & =\frac{2}{K}\operatorname{Re}\left\{ {{e}^{-j\pi \left( {{K}_{\text{all}}}-1 \right){{T}_{0}}{{f}_{d}}}}\sum\limits_{k=0}^{\frac{K}{2}-1}{{{e}^{j2\pi {{\left[ {\bm{\tilde{\varphi }}} \right]}_{k}}{{T}_{0}}{{f}_{d}}}}} \right\} \\
\nonumber  & =\operatorname{Re}\left\{ {{e}^{-j\pi \left( {{K}_{\text{all}}}-1 \right){{T}_{0}}{{f}_{d}}+j\angle \left\{ {\mathcal{E}_\mathcal{P}}\left( \bm{\tilde{\varphi }},{{f}_{d}} \right) \right\}}} \right\}{\mathcal{E}_\mathcal{P}}\left( \bm{\tilde{\varphi }},{{f}_{d}} \right)\\
 & =\cos \left( \pi \left( {{K}_{\text{all}}} \! - \! 1 \right){{T}_{0}}{{f}_{d}} \! - \! \angle \left\{ {\mathcal{E}_\mathcal{P}}\left( \bm{\tilde{\varphi }},{{f}_{d}} \right) \right\} \right){\mathcal{E}_\mathcal{P}}\left( \bm{\tilde{\varphi }},{{f}_{d}} \right),
\end{align}
where
${\mathcal{E}_\mathcal{P}}\left( {\bm{\tilde{\varphi} }}, f_d \right)$
is
is the envelope of
$\mathcal{P}\left( {\bm{\varphi }}, f_d \right)$,
and
{$\mathcal{P}\left( {\bm{\varphi }}, f_d \right)$}
is modulated by
$\cos \left( \pi \left( {{K}_{\text{all}}}-1 \right){{T}_{0}}{{f}_{d}}-\angle \left\{ {{\mathcal{P}}}\left( \bm{\tilde{\varphi }},{{f}_{d}} \right) \right\} \right)$,
i.e., a ``high-frequency oscillating signal''.
Proposition \ref{Lemma_Envelope_in_aoifhnaol} is proved.

\section{Proof of Proposition \ref{proposition_integer_operation}}\label{Appendix_of_Proposition_Integer_Operation}
Suppose that
${\bm{{\tilde{\varphi }}}_{B}^{\prime}}={{{\bm{\tilde{\varphi }}}}_{B}}+{{\bm{\delta }}_{B}}$,
where
${{\bm{\delta }}_{B}} \in \mathbb{R}^{K_B\times 1}$
and
$\left| {{\bm{\delta }}_{B}} \right|\le \frac{1}{2}{{\mathbf{1}}_{{{K}_{B}}\times 1}}$.
We have:
\begin{align}
\nonumber  & \left| {{\mathcal{E}}_{\mathcal{P}}}\left( \left[ \begin{matrix}
   {{{\bm{\tilde{\varphi }}}}_{F}}  \\
   {{{\bm{\tilde{\varphi }}}}_{B}}  \\
\end{matrix} \right],{{f}_{d}} \right)-{{\mathcal{E}}_{\mathcal{P}}}\left( \left[ \begin{matrix}
   {{{\bm{\tilde{\varphi }}}}_{F}}  \\
   {{{\bm{\tilde{\varphi }}}}_{B}}+{{\bm{\delta }}_{B}}  \\
\end{matrix} \right],{{f}_{d}} \right) \right| \\
\nonumber   & \le \left| \frac{2}{K}\sum\nolimits_{k=0}^{{{K}_{B}}-1}{{{e}^{j2\pi {{\left[ {{{\bm{\tilde{\varphi }}}}_{B}} \right]}_{k}}{{T}_{0}}{{f}_{d}}}}\left\{ 1-{{e}^{j2\pi {{\left[ {{\bm{\delta }}_{B}} \right]}_{k}}{{T}_{0}}{{f}_{d}}}} \right\}} \right| \\
\nonumber   & \approx \left| \frac{2}{K}2\pi {{T}_{0}}{{f}_{d}}\sum\nolimits_{k=0}^{{{K}_{B}}-1}{{{\left[ {{\bm{\delta }}_{B}} \right]}_{k}}{{e}^{j2\pi {{\left[ {{{\bm{\tilde{\varphi }}}}_{B}} \right]}_{k}}{{T}_{0}}{{f}_{d}}}}} \right| \\
 & <\left| \frac{2}{K}2\pi {{T}_{0}}{{f}_{d}}\sum\nolimits_{k=0}^{{{K}_{B}}-1}{\frac{1}{2}} \right|=\frac{\pi {{K}_{B}}}{K/2}\left| {{T}_{0}}{{f}_{d}} \right|.
\end{align}
Proposition \ref{proposition_integer_operation} is thus proved.

\section{Proof of (\ref{iteration_algorithm_result})}\label{Appendix_iteration_algorithm_result}
When $\mathcal{H}_{\text{op}}^{\left( i+1 \right)}$ is minimized,
the partial derivative of $\mathcal{H}_{\text{op}}^{\left( i+1 \right)}$
with respect to $\Delta \bm{\tilde{\varphi }}_{B}^{\left( i+1 \right)} $
is zero, i.e.,
$\frac{\partial }{\partial \Delta \bm{\tilde{\varphi }}_{B}^{\left( i+1 \right)}}\mathcal{H}_{\text{op}}^{\left( i+1 \right)}=0$,
which leads to:
\begin{align}\label{Appendix_formula_aldiasoshf}
\frac{\partial }{\partial \Delta \bm{\tilde{\varphi }}_{B}^{\left( i+1 \right)}}\mathcal{P}_{\text{SL}}^{\left( i+1 \right)}+2a\text{diag}\left\{ \mathbf{w}^{\left( i \right)} \right\}\Delta \bm{\tilde{\varphi }}_{B}^{\left( i+1 \right)}=0.
\end{align}
Given the relationship between ${{\mathcal{P}}_{\text{SL}}}$
and
${\mathcal{E}_\mathcal{P}}\left( {{{\bm{\tilde{\varphi }}}}},{{f}_{d}} \right)$ in (\ref{weighted_sidelobe_power_aasdlka}), the partial derivative of ${{\left| {\mathcal{E}_\mathcal{P}}\left( {{{\bm{\tilde{\varphi }}}}^{\left( i+1 \right)}},{{f}_{d}} \right) \right|}^{2}}$
with respect to $\Delta \bm{\tilde{\varphi }}_{B}^{\left( i+1 \right)} $
is needed for deriving the first term on the left-hand side above.
Based on (\ref{weighted_sidelobe_power_aasdlka}),
we can show that:
\begin{align}\label{Appendix_formula_aldiadasfasoshf}
\nonumber  & \frac{\partial}{\partial \Delta \bm{\tilde{\varphi }}_{B}^{\left( i+1 \right)}}{{\left| {\mathcal{E}_\mathcal{P}}\left( {{{\bm{\tilde{\varphi }}}}^{\left( i+1 \right)}},{{f}_{d}} \right) \right|}^{2}} \\
\nonumber & =-2{{f}_{d}}\frac{8\pi {{T}_{0}}}{{{K}^{2}}}\operatorname{Im}\left\{ \mathbf{f}\left( \bm{\tilde{\varphi }}_{B}^{\left( i \right)},{{f}_{d}} \right){{\mathbf{f}}^{H}}\left( \bm{\tilde{\varphi }},{{f}_{d}} \right) \right\}{{\mathbf{1}}_{{K}/{2}\;\times 1}} \\
 & +2f_{d}^{2}\frac{4{{\left( 2\pi {{T}_{0}} \right)}^{2}}}{{{K}^{2}}}\operatorname{Re}\left\{ \mathbf{f}\left( \bm{\tilde{\varphi }}_{B}^{\left( i \right)},{{f}_{d}} \right){{\mathbf{f}}^{H}}\left( \bm{\tilde{\varphi }}_{B}^{\left( i \right)},{{f}_{d}} \right) \right\}\Delta \bm{\tilde{\varphi }}_{B}^{\left( i+1 \right)}.
\end{align}
Substituting (\ref{Appendix_formula_aldiadasfasoshf}) into the first term of (\ref{Appendix_formula_aldiasoshf}) yields:
\begin{align}\label{Appendix_formula_aldiadasfafdasasoshf}
\nonumber \frac{\partial }{\partial \Delta \bm{\tilde{\varphi }}_{B}^{\left( i+1 \right)}}\mathcal{P}_{\text{SL}}^{\left( i+1 \right)}  & =\int_{f_{d,\text{M}}^{\text{ex}}}^{\frac{1}{2{{T}_{0}}}}{\varpi \left( {{f}_{d}} \right)\frac{\partial{{\left| {\mathcal{E}_\mathcal{P}}\left( {{{\bm{\tilde{\varphi }}}}^{\left( i+1 \right)}},{{f}_{d}} \right) \right|}^{2}}}{\partial \Delta \bm{\tilde{\varphi }}_{B}^{\left( i+1 \right)}}d{{f}_{d}}} \\
 & \approx 2\mathbf{V}_{2}^{\left( i \right)}\Delta \bm{\tilde{\varphi }}_{B}^{\left( i+1 \right)}-2\mathbf{V}_{1}^{\left( i \right)}{{\mathbf{1}}_{{K}/{2}\;\times 1}},
\end{align}
where the definition of $\mathbf{V}_{1}^{\left( i \right)}\in {{\mathbb{C}}^{{{K}_{B}}\times \frac{K}{2}}}$
and
$\mathbf{V}_{2}^{\left( i \right)}\in {{\mathbb{C}}^{{{K}_{B}}\times {{K}_{B}}}}$
is given in (\ref{iteration_algorithm_result}).
{Substituting} (\ref{Appendix_formula_aldiadasfafdasasoshf}) into (\ref{Appendix_formula_aldiasoshf}),
the formula (\ref{iteration_algorithm_result}) can be proved.

\end{appendices}

\bibliographystyle{IEEEtran}
\bibliography{AAA}
\vfill

\end{document}